\documentclass[journal]{IEEEtran} \newcommand{\figsize}{0.9\columnwidth}
\usepackage{graphicx,psfrag}
\usepackage{url,hyperref}
\usepackage[usenames]{color}
\usepackage[cmex10]{amsmath}
\usepackage[nospace]{cite}
\usepackage{amsfonts,amsthm,amssymb,latexsym,cite,color,multirow,rotating}
\usepackage{algorithm}
\usepackage[noend]{algpseudocode}
\usepackage{bbm}
\usepackage[normalem]{ulem} 
\usepackage{subcaption}
\usepackage{xspace} 
\usepackage{xpatch} 

 \newcommand{\putTable}[3]{\begin{table}[t]
  			    \centering
		            #3
     			    \caption{\small #2}
     			    \label{tab:#1}
			    \vspace{-4mm}
			  \end{table} }
			  
 \newcommand{\putFrag}[4]{\begin{figure}[t]
                            \begin{center}
                            #4
			    \includegraphics[width=#3]{figures/#1.eps}
            		    \caption{\small #2}
           		    \label{fig:#1}
                            \end{center}
                          \end{figure} }

 \newcommand{\twoFrag}[6]{\begin{figure}[t]
                            \begin{center}
                            #6
			    \hfill
			    \includegraphics[width=#3\columnwidth]{figures/#1.eps}
			    \hfill
			    \includegraphics[width=#4\columnwidth]{figures/#2.eps}
			    \hfill
                            \mbox{}
            		    \caption{\small #5}
           		    \label{fig:#1}
                            \end{center}
                          \end{figure} }
                          
 \newcommand{\threeFrag}[8]{\begin{figure}[t]
                            \begin{center}
                            #8
			    \hfill
			    \includegraphics[width=#4\columnwidth]{figures/#1.eps}
			    \hfill
			    \includegraphics[width=#5\columnwidth]{figures/#2.eps}
			    \hfill
			    \includegraphics[width=#6\columnwidth]{figures/#3.eps}
			    \hfill
                            \mbox{}
            		    \caption{\small #7}
           		    \label{fig:#1}
                            \end{center}
                          \end{figure} }
                          
 \newcommand{\defn}{\triangleq}

 \newcommand{\hvec}[1]{\ensuremath{\boldsymbol{\Hat{#1}}}}
    
 \renewcommand{\vec}[1]{\ensuremath{\boldsymbol{#1}}}
 \newcommand{\mat}[1]{\ensuremath{\begin{bmatrix}#1\end{bmatrix}}}

 \newcommand{\norm}[1]{\ensuremath{\| #1 \|}}
 \newcommand{\mc}[1]{\ensuremath{\mathcal{#1}}}
 \newcommand{\st}{{~\text{s.t.}~}}
 \newcommand{\barst}[1]{\ensuremath{\text{\raisebox{-0.5mm}{$\big|_{#1}$}}}}

 \newcommand{\Real}{{\mathbb{R}}}
 \newcommand{\Complex}{{\mathbb{C}}}

 \newcommand{\tran}{^\textsf{T}}
 
 \newcommand{\of}[1]{^{(#1)}}

 \DeclareMathOperator{\diag}{diag}

 \DeclareMathOperator*{\argmax}{arg\, max}
 \DeclareMathOperator*{\argmin}{arg\, min}

 \newcommand*\deriv{\mathop{}\!\mathrm{d}}

 \newtheorem{theorem}{Theorem}
 \newtheorem{corollary}[theorem]{Corollary} 

 \renewcommand{\eqref}[1]{(\ref{eq:#1})}
 
 \newcommand{\Figref}[1]{Figure~\ref{fig:#1}}
 
 \newcommand{\figref}[1]{Fig.~\ref{fig:#1}}
 
 \newcommand{\tabref}[1]{Table~\ref{tab:#1}}
 \newcommand{\secref}[1]{Sec.~\ref{sec:#1}}

 \newcommand{\appref}[1]{Appendix~\ref{app:#1}}
 
 \newcommand{\thmref}[1]{Theorem~\ref{thm:#1}}
 \newcommand{\corref}[1]{Corollary~\ref{cor:#1}}

 \renewcommand{\algref}[1]{Algorithm~\ref{alg:#1}}
 
 \newcommand{\lineref}[1]{line~\ref{line:#1}}
 \newcommand{\partref}[1]{Part~\ref{part:#1}}

 \newcounter{comment}[section]
 
 \newcounter{texthead}[section]

 \newcommand{\map}{_\textsf{MAP}}
 
 \newcommand{\IRW}{IRW-L1\xspace}
 \newcommand{\score}{Co-L1\xspace}
 \newcommand{\scoreirw}{Co-IRW-L1-\texorpdfstring{$\vec{\epsilon}$\xspace}{eps}}
 \newcommand{\ascoreirw}{Co-IRW-L1\xspace}
 \newcommand{\Rell}{R_1}
 \newcommand{\Rlog}{R_{\text{ls}}}
 \newcommand{\Rlsl}{R_{\text{lsl}}}
 \newcommand{\const}{\text{const}}
 \newcommand{\Dkl}{D_{\text{\textsf{KL}}}}
 \newcommand{\mysf}[1]{\textsf{\small{#1}}}

\makeatletter
\def\BState{\State\hskip-\ALG@thistlm}
\xpatchcmd{\algorithmic}{\itemsep\z@}{\itemsep=0.5ex plus2pt}{}{} 
\makeatother

\begin{document}

\setlength{\arraycolsep}{0.4mm}
 \title{Iteratively Reweighted \texorpdfstring{$\ell_1$}{L1} Approaches to Sparse Composite Regularization}
         \author{Rizwan Ahmad and Philip Schniter, \emph{Fellow, IEEE}%
         \thanks{This work has been supported in part by NSF grants CCF-1218754 and CCF-1018368. Portions of this work were presented at the 2015 ISMRM Annual Meeting and Exhibition.}%
         \thanks{The authors are with the Department of Electrical and Computer Engineering, The Ohio State University, Columbus, OH, USA.  (rizwan.ahmad@osumc.edu, schniter@ece.osu.edu, phone 614.247.6488, fax 614.292.7596).}%
         }
 \date{\today}
 \maketitle

\begin{abstract}
Motivated by the observation that a given signal $\boldsymbol{x}$ admits sparse representations in multiple dictionaries $\boldsymbol{\Psi}_d$ but with varying levels of sparsity across dictionaries, we propose two new algorithms for the reconstruction of (approximately) sparse signals from noisy linear measurements.
Our first algorithm, Co-L1, extends the well-known lasso algorithm from the L1 regularizer $\|\boldsymbol{\Psi x}\|_1$ to composite regularizers of the form $\sum_d \lambda_d \|\boldsymbol{\Psi}_d \boldsymbol{x}\|_1$ while self-adjusting the regularization weights $\lambda_d$.
Our second algorithm, Co-IRW-L1, extends the well-known iteratively reweighted L1 algorithm to the same family of composite regularizers.
We provide several interpretations of both algorithms: 
i) majorization-minimization (MM) applied to a non-convex log-sum-type penalty,
ii) MM applied to an approximate $\ell_0$-type penalty,
iii) MM applied to Bayesian MAP inference under a particular hierarchical prior, and 
iv) variational expectation-maximization (VEM) under a particular prior with deterministic unknown parameters.
A detailed numerical study suggests that our proposed algorithms yield significantly improved recovery SNR when compared to their non-composite L1 and IRW-L1 counterparts.
\end{abstract}

\section{Introduction} \label{sec:intro}
We consider the problem of recovering the signal (or image) $\vec{x} \in \Complex^N$ from noisy linear measurements of the form
\begin{align}
\vec{y} 
&= \vec{\Phi x}+\vec{w} \in \Complex^M,
\label{eq:y}
\end{align}
where $\vec{\Phi} \in \Complex^{M \times N}$ is a known measurement operator and
$\vec{w} \in \Complex^M$ is additive noise.
Such problems arise in imaging, machine learning, radar, communications, speech, and many other applications.
We are particularly interested in the case that $M\ll N$, where $\vec{x}$ cannot be uniquely determined from the measurements $\vec{y}$, even in the absence of noise. 
This latter situation arises in many of the aforementioned applications, as well as in broad area of signal recovery methods associated with \emph{compressive sensing} (CS) \cite{CandesWakin:SPM:08}.

\subsection{Regularized \texorpdfstring{$\ell_2$}{L2}-Minimization} \label{sec:opt}
By incorporating (partial) prior knowledge about the signal and noise power, it may be possible to accurately recover $\vec{x}$ from $M\ll N$ measurements $\vec{y}$.
In this work, we consider signal recovery based on optimization problems of the form
\begin{align}
\argmin_{\vec{x}} 
\gamma \|\vec{y}-\vec{\Phi x}\|^2_2  +
R(\vec{x}) 
\label{eq:opt} 
\end{align}
where $\gamma$ is a tuning parameter that reflects knowledge of the noise level
and $R(\vec{x})$ is a penalty, or regularization, that reflects prior knowledge about the signal $\vec{x}$ \cite{Mairal:FTCV:14}.
We briefly summarize several common instances of $R(\vec{x})$ below.
\begin{enumerate}[\IEEEsetlabelwidth{0}]
\item
If $\vec{x}$ is known to be \emph{sparse} (i.e., contains sufficiently few non-zero coefficients) or approximately sparse, then one would ideally like to use 
the $\ell_0$ penalty (i.e., counting ``norm'') $R(\vec{x})=\|\vec{x}\|_0\defn\sum_{n=1}^N 1_{|x_n|>0}$, where $1_{\{\cdot\}}$ is the indicator function.
However, since this choice makes \eqref{opt} NP-hard, it is not often used in practice.
\item
The $\ell_1$ penalty, $R(\vec{x})\!=\!\|\vec{x}\|_1\!=\!\sum_{n=1}^N|x_n|$, is a well-known surrogate to the $\ell_0$ penalty that renders \eqref{opt} convex, and thus amenable to polynomial-time solution.
In this case, \eqref{opt} is known as the \emph{basis pursuit denoising} \cite{Chen:JSC:98} or \emph{lasso} \cite{Tibshirani:JRSSb:96} problem, which is commonly used in \emph{synthesis-based} CS \cite{CandesWakin:SPM:08}.
\item
Various non-convex surrogates for the $\ell_0$ penalty have also been considered, such as the $\ell_p$ penalty $R(\vec{x})\!=\!\|\vec{x}\|_p^p\!=\!\sum_{n=1}^N|x_n|^p$ with $p\in(0,1)$ and the log-sum penalty $R(\vec{x})\!=\!\sum_{n=1}^N \log(\epsilon+|x_n|)$ with $\epsilon\geq 0$. 
Although \eqref{opt} becomes difficult to solve, it can be tractably approximated.
See \cite{Mairal:FTCV:14} for a more complete discussion.
\item
The choice $R(\vec{x})=\|\vec{\Psi x}\|_1$, with known matrix $\vec{\Psi}\in\Complex^{L\times N}$, leads to \emph{analysis-based} CS \cite{Elad:IP:07} and the \emph{generalized lasso} \cite{Tibshirani:AS:11}.
Penalties of this form are appropriate when prior knowledge suggests that the transform 
coefficients $\vec{\Psi x}$ are (approximately) sparse, as opposed to the signal $\vec{x}$ itself being sparse.
When $\vec{\Psi}$ is a finite-difference operator, $\|\vec{\Psi x}\|_1$ yields anisotropic \emph{total variation} regularization \cite{Rudin:PhyD:92}.
\item
Non-convex penalties can also be placed on the transform coefficients $\vec{\Psi x}$, leading to, e.g.,
$R(\vec{x})\!=\!\|\vec{\Psi x}\|_p^p\!=\!\sum_{l=1}^L|\vec{\psi}_l\tran\vec{x}|^p$ with $p\in(0,1)$ or $R(\vec{x})\!=\!\sum_{l=1}^L \log(\epsilon+|\vec{\psi}_l\tran\vec{x}|)$ with $\epsilon\geq 0$. 

\end{enumerate}

A popular approach to solve \eqref{opt} with a non-convex penalty $R(\vec{x})$ is through \emph{iteratively reweighted $\ell_1$} (\IRW)\footnote{Iteratively reweighted $\ell_2$ is a popular alternative, e.g., \cite{Figueiredo:TPAMI:03,Figueiredo:TIP:07,Chartrand:ICASSP:08,Daubechies:CPAM:10,Wipf:JSTSP:10}.}
\cite{Figueiredo:TIP:07}.
There, \eqref{opt} with fixed non-convex $R(\vec{x})$ is approximated by solving a sequence of convex problems
\begin{align}
\vec{x}\of{t}
&= \argmin_{\vec{x}} 
\gamma \|\vec{y}-\vec{\Phi x}\|^2_2  + R\of{t}(\vec{x}) ,
\end{align}
where, at iteration $t$, the penalty $R\of{t}(\vec{x})\!=\!\sum_{n=1}^N w_n\of{t}|x_n|$ with each weight $w_n\of{t}$ set based on the previous estimate $x_n\of{t-1}$.
Constrained formulations of \IRW based on ``$\vec{x}\of{t}\!=\!\arg\min_{\vec{x}} R\of{t}(\vec{x}) \st \|\vec{y}-\vec{\Phi x}\|_2\leq \delta$,'' have also been considered, such as in \cite{Candes:JFA:08,Wipf:JSTSP:10,Carrillo:SPL:13}.
Many of the papers cited above show empirical results where the performance of \IRW surpasses that of standard $\ell_1$.

\subsection{Sparsity-Inducing Composite Regularizers} \label{sec:composite}

In this work, we focus on sparsity-inducing \emph{composite} regularizers of the form 
\begin{align}
\Rell^D(\vec{x};\vec{\lambda})
&\defn \sum_{d=1}^D \lambda_d \|\vec{\Psi}_d \vec{x}\|_1 
\label{eq:Rell} ,
\end{align}
where each $\vec{\Psi}_d\in\Complex^{L_d\times N}$ is a known analysis operator and $\lambda_d\geq 0$ is a corresponding regularization weight.
Our goal is to recover the signal $\vec{x}$ from measurements \eqref{y} by optimizing \eqref{opt} with the composite regularizer \eqref{Rell}.
Doing so requires an optimization of the weights $\vec{\lambda}=[\lambda_1,\dots,\lambda_D]\tran$ in \eqref{Rell}.
We are also interested in iteratively re-weighted extensions of this problem that, at iteration $t$, use composite regularizers of the form\footnote{Although \eqref{Rellirw} is over-parameterized, the form of \eqref{Rellirw} is convenient for algorithm development.} 
\begin{align}
R\of{t}(\vec{x})
&= \sum_{d=1}^D \lambda_d\of{t} \|\vec{W}_d\of{t}\vec{\Psi}_d \vec{x}\|_1 ,
\label{eq:Rellirw}  
\end{align}
where $\vec{W}_d\of{t}$ are diagonal matrices.
This latter approach requires the optimization of both $\lambda_d\of{t}$ and $\vec{W}_d\of{t}$ for all $d$.

As a motivating example, 
suppose that $\{\vec{\Psi}_d\}$ is a collection of orthonormal bases that includes, e.g., spikes, sines, and various wavelet bases.
The signal $\vec{x}$ may be sparse in some of these bases, but not all.  
Thus, we would like to adjust each $\lambda_d$ in \eqref{Rell} to appropriately weight the contribution from each basis.  
But it is not clear how to do this, especially since $\vec{x}$ is unknown.
As another example, 
suppose that $\vec{x}$ contains a (rasterized) sequence of images and that $\|\vec{\Psi}_1\vec{x}\|_1$ measures temporal total-variation while $\|\vec{\Psi}_2\vec{x}\|_1$ measures spatial total-variation.
Intuitively, we would like to weight these two regularizations differently, depending on whether the image pixels vary more in the temporal or spatial dimensions.
But it is not clear how to do this, especially since $\vec{x}$ is unknown.

\subsection{Contributions}

In this work, we propose novel iteratively reweighted approaches to sparse reconstruction based on composite regularizations of the form \eqref{Rell}-\eqref{Rellirw} with automatic tuning of the regularization weights $\vec{\lambda}$ and $\vec{W}_d$.
For each of our proposed algorithms, we will provide four interpretations:
\begin{enumerate}
\item
MM applied to a non-convex log-sum-type penalty,
\item
MM applied to an approximate $\ell_0$-type penalty,
\item
MM applied to Bayesian MAP inference based on Gamma and Jeffrey's hyperpriors \cite{Figueiredo:TIP:01,Oliveira:SP:09}, 
and 
\item
variational expectation maximization (VEM) \cite{Neal:Jordan:98,Bishop:Book:07} applied to a Laplacian or generalized-Pareto prior with deterministic unknown parameters.
\end{enumerate}
We show that the MM interpretation guarantees convergence in the sense of satisfying an asymptotic stationary point condition \cite{Mairal:ICML:13}.
Moreover, we establish connections between our proposed approaches and existing \IRW algorithms, and we provide novel VEM-based and Bayesian MAP interpretations of those existing algorithms.

Finally, through the detailed numerical study in \secref{num}, we establish that our proposed algorithms yield significant gains in recovery accuracy relative to existing methods with only modest increases in runtime.
In particular, when $\{\vec{\Psi}_d\}$ are chosen so that the sparsity of $\vec{\Psi}_d\vec{x}$ varies with $d$, this structure can be exploited for improved recovery.  The more disparate the sparsity, the greater the improvement.

\subsection{Related Work} \label{sec:related}

As discussed above, the generalized lasso \cite{Tibshirani:AS:11} is one of the most common approaches to L1-regularized analysis-CS \cite{Elad:IP:07}, i.e., the optimization \eqref{opt} under the regularizer $R(\vec{x})=\|\vec{\Psi x}\|_1$. 
The \score algorithm that we present in \secref{score} can be interpreted as a generalization of this L1 method to \emph{composite} regularizers of the form \eqref{Rell}.
Meanwhile, the iteratively reweighted extension of the generalized lasso, \IRW \cite{Figueiredo:TIP:07}, often yields significantly better reconstruction accuracy with a modest increase in complexity (e.g., \cite{Candes:JFA:08,Carrillo:SPL:13}).
The \ascoreirw algorithm that we present in \secref{scoreirw} can be interpreted as a generalization of this \IRW method to \emph{composite} regularizers of the form \eqref{Rellirw}.
The existing non-composite L1 and \IRW approaches essentially place an identical weight $\lambda_d=1$ on every term in \eqref{Rell}-\eqref{Rellirw}, and thus make no attempt to leverage differences in the sparsity of the transform coefficients $\vec{\Psi}_d\vec{x}$ across the sub-dictionary index $d$.
However, the numerical results that we present in \secref{num} suggest that there can be significant advantages to optimizing $\lambda_d$, which is precisely what our methods do.

The problem of optimizing the weights $\lambda_d$ of composite regularizers $R(\vec{x};\vec{\lambda})=\sum_d \lambda_d R_d(\vec{x})$ is a long-standing problem with a rich literature (see, e.g., the recent book \cite{Lu:Book:13}). 
However, the vast majority of that literature focuses on the Tikhonov case where $R_d(\vec{x})$ are quadratic (see, e.g., \cite{Brezinski:NM:03,Xu:JG:06,Gazzola:ETNA:13,Fornasier:JNA:14}).
One notable exception is \cite{Belge:IP:02}, which assumes continuously differentiable $R_d(\vec{x})$ and thus does not cover our composite $\ell_1$ prior \eqref{Rell}.
Another notable exception is \cite{Kunisch:JIS:13}, which assumes i) the availability of a noiseless training example of $\vec{x}$ to help tune the L1 regularization weights $\vec{\lambda}$ in \eqref{Rell}, and ii) the trivial measurement matrix $\vec{\Phi}=\vec{I}$.
In contrast, our proposed methods operate without any training and support generic measurement matrices $\vec{\Phi}$.

In the special case that each $\vec{\Psi}_d$ is composed of a subset of rows from the $N\times N$ identity matrix,
the regularizers \eqref{Rell}-\eqref{Rellirw} can induce \emph{group} sparsity in the recovery of $\vec{x}$, in that certain sub-vectors $\vec{x}_d\defn\vec{\Psi}_d\vec{x}$ of $\vec{x}$ are driven to zero while others are not.
The paper \cite{Rakotomamonjy:SP:11} develops an \IRW-based approach to group-sparse signal recovery for equal-sized non-overlapping groups that can be considered as a special case of the \score algorithm that we develop in \secref{score}.
However, our approach is more general in that it handles possibly non-equal and/or overlapping groups, not to mention sparsity in a generic set of sub-dictionaries $\vec{\Psi}_d$.
Recently, Bayesian MAP group-sparse recovery was considered in \cite{Babacan:TSP:14}.
However, the technique described there uses Gaussian scale mixtures or, equivalently, weighted-L2 regularizers $R(\vec{x};\vec{\lambda})=\sum_d \lambda_d \|\vec{x}_d\|_2$, while our methods use weighted-$\ell_1$ regularizers \eqref{Rell}-\eqref{Rellirw}.

\subsection{Notation}

We use boldface capital letters like $\vec{\Psi}$ for matrices,
boldface small letters like $\vec{x}$ for vectors,
and $(\cdot)\tran$ for transposition.
We use $\norm{\vec{x}}_p=(\sum_n |x_n|^p)^{1/p}$ for the $\ell_p$ norm of $\vec{x}$, with $x_n$ representing the $n^{th}$ coefficient in $\vec{x}$ and $p>0$.
We then use $\norm{\vec{x}}_0=\lim_{p\rightarrow 0} \sum_n |x_n|^p$ \cite{Wipf:JSTSP:10} when referring to the $\ell_0$ quasi-norm, which counts the number of nonzero coefficients in $\vec{x}$. 
We define the ``mixed $\ell_{p,0}$ quasi-norm'' with $p>0$ as\footnote{%
Our $\ell_{p,0}$ and $\ell_{0,0}$ definitions are motivated by the standard $\ell_{p,q}$ mixed norm definition (for $p,q>0$), which is $(\sum_d (\sum_l |x_{d,l}|^p)^{q/p})^{1/q}$ \cite{Kowalski:ACHA:09}.}
$\lim_{q\rightarrow 0} \sum_d (\sum_l |x_{d,l}|^p)^q$, 
and the ``mixed $\ell_{0,0}$ quasi-norm'' as $\lim_{p,q\rightarrow 0} \sum_d (\sum_l |x_{d,l}|^p)^q$. 
We use $\nabla g(\vec{x})$ for the gradient of a functional $g(\vec{x})$ with respect to $\vec{x}$, and 
$1_A$ for the indicator function that returns the value $1$ when $A$ is true and $0$ when $A$ is false.
We use $p(\vec{x};\vec{\lambda})$ for the pdf of random vector $\vec{x}$ under deterministic parameters $\vec{\lambda}$, and $p(\vec{x}|\vec{\lambda})$ for the pdf of $\vec{x}$ conditioned on the random vector $\vec{\lambda}$.
We use $\Dkl(q\|p)$ to denote the Kullback-Leibler (KL) divergence of pdf $p$ from pdf $q$, 
and 
we use $\Real$ and $\Complex$ to denote the real and complex fields, respectively.

\section{The \score Algorithm} \label{sec:score}

We first propose the Composite-L1 (\score) algorithm, which is summarized in \algref{score}.
There, $L_d$ denotes the number of rows in $\vec{\Psi}_d$.

\begin{algorithm}[h]
  \caption{The \score Algorithm}
  \label{alg:score}
  \begin{algorithmic}[1]
    \State
    \textsf{input:~~} $\{\vec{\Psi}_d\}_{d=1}^D$, $\vec{\Phi}$, $\vec{y}$, 
    $\gamma>0$,
    $\epsilon \geq 0$\\
    \mysf{if} $\vec{\Psi}_d\vec{x}\in\Real^{L_d}$, \mysf{use} $C_d=1$;
    \mysf{if} $\vec{\Psi}_d\vec{x}\in\Complex^{L_d}$, \mysf{use} $C_d=2$.\\
    \textsf{initialization:~~} $\lambda_d\of{1}=1~\forall d$
    \State
    \textsf{for~} $t=1,2,3,\dots$\\
    \label{line:score_x}
    \quad $\vec{x}\of{t} \gets \displaystyle \argmin_{\vec{x}} 
        \gamma \|\vec{y} - \vec{\Phi x}\|^2_2 +
        \sum_{d=1}^D\lambda_d\of{t} \|\vec{\Psi}_d\vec{x}\|_1$
    \\
    \label{line:score_lam}
    \quad $\displaystyle \lambda_d\of{t+1} \gets \frac{C_d L_d}{\epsilon+\|\vec{\Psi}_d \vec{x}\of{t}\|_1},~~d=1,\dots,D$\\
    \textsf{end}\\
     \textsf{output:~~}$\vec{x}\of{t}$
  \end{algorithmic}
\end{algorithm}

The main computational step of \score is the L2+L1 minimization in \lineref{score_x}, which can be recognized as \eqref{opt} under the composite regularizer $\Rell^D$ from \eqref{Rell}.
This is a convex optimization problem that can be readily solved by existing techniques (e.g., 
ADMM \cite{Boyd:FTML:10,Afonso:TIP:10},
Douglas-Rachford splitting \cite{Combettes:JSTSP:07}, 
MFISTA \cite{Tan:TSP:14}, 
NESTA-UP \cite{Becker:JIS:11}, 
GAMP \cite{Borgerding:ICASSP:15},
etc.),
the specific choice of which is immaterial to this paper.

Note that 
\score requires the user to set a small regularization term $\epsilon\geq 0$ whose role is to prevent the denominator in \lineref{score_lam} from reaching zero.
For typical choices of $\vec{\Psi}_d$ and $\gamma$, the vector $\vec{\Psi}_d\vec{x}\of{t}$ will almost never be exactly zero, in which case it suffices to set $\epsilon=0$.
Also, \score requires the user to set the measurement fidelity weight $\gamma$. 
With additive white Gaussian noise (AWGN) of variance $\sigma^2>0$, the Bayesian MAP interpretation discussed in \secref{score_bayes} suggests setting $\gamma=\frac{1}{2\sigma^2}$ for real-valued AWGN or $\gamma=\frac{1}{\sigma^2}$ for circular complex-valued AWGN. 
These are, in fact, the settings that we used for all numerical results in \secref{num}.
 
Note \lineref{score_x} of \algref{score} can be equivalently restated as
\begin{align}
\vec{x}\of{t} 
&\gets \displaystyle \argmin_{\vec{x}} 
\sum_{d=1}^D\lambda_d\of{t} \|\vec{\Psi}_d\vec{x}\|_1 
        \st \|\vec{y} - \vec{\Phi x}\|_2 \leq \delta .
\label{eq:score_x_con}
\end{align}
By equivalent, we mean that, for any $\delta>0$, there exists a $\gamma$ for which the solutions of \lineref{score_x} and \eqref{score_x_con} are identical \cite{Lorenz:IP:13}.
A version of this manuscript that focuses on the constrained case can be found at \cite{Ahmad:constrained:15}.
Numerical experiments therein show that the performance of \score using \eqref{score_x_con} with the hand-tuned value $\delta=0.8\sqrt{M \sigma^2}$ is very similar to that of \algref{score} with $\gamma$ chosen as described above.

\score's update of the weights $\vec{\lambda}$, defined by \lineref{score_lam} of \algref{score}, can be interpreted in various ways, as we detail below.
For ease of explanation, we first consider the case where $\vec{\Psi}_d\vec{x}$ is real-valued $\forall d$, and later discuss the complex-valued case in \secref{score_complex}.

\begin{theorem}[\score] \label{thm:score}
The \score algorithm in \algref{score} has the following interpretations:
\begin{enumerate}
\item \label{part:score_logsum}
MM applied to \eqref{opt} under the log-sum penalty 
\begin{align}
\Rlog^D(\vec{x};\epsilon)
&\defn \sum_{d=1}^D L_d \log( \epsilon +\|\vec{\Psi}_d \vec{x}\|_1 ) ,
\label{eq:Rlog}
\end{align}

\item \label{part:score_ell0}
as $\epsilon\rightarrow 0$,
MM applied to \eqref{opt} under the weighted $\ell_{1,0}$ \cite{Kowalski:ACHA:09} penalty
\begin{align}
R_{10}^D(\vec{x})
&\defn \sum_{d=1}^D L_d \,1_{\|\vec{\Psi}_d\vec{x}\|_1>0}
\label{eq:score_ell0} ,
\end{align}

\item \label{part:score_bayes}
MM applied to Bayesian MAP estimation under an additive white Gaussian noise (AWGN) likelihood and the hierarchical prior
\begin{align}
p(\vec{x}|\vec{\lambda}) 
&= \prod_{d=1}^D \bigg(\frac{\lambda_d}{2}\bigg)^{L_d} \exp\big({-\lambda_d} \|\vec{\Psi}_d\vec{x}\|_1\big)
\label{eq:score_bayes_prior_x} \\
\vec{\lambda} &\sim \text{i.i.d.~}\Gamma(0,\epsilon^{-1})
\label{eq:score_bayes_prior_lam} 
\end{align}
where $\vec{z}_d\!\defn\!\vec{\Psi}_d\vec{x}\in\Real^{L_d}$ is i.i.d.\ Laplacian given $\lambda_d$, and $\lambda_d$ is Gamma distributed with scale parameter $\epsilon^{-1}$ and shape parameter zero, which becomes
Jeffrey's non-informative hyperprior $p(\lambda_d)\propto 1_{\lambda_d>0}/\lambda_d$ when $\epsilon=0$.  
\item \label{part:score_em}
variational EM under an AWGN likelihood and the prior
\begin{align}
p(\vec{x};\vec{\lambda}) 
&\propto \prod_{d=1}^D \bigg(\frac{\lambda_d}{2}\bigg)^{L_d} \!\! \exp\big({-\lambda_d} (\|\vec{\Psi}_d\vec{x}\|_1+\epsilon)\big) ,
\label{eq:score_em_prior_x} 
\end{align}
which, when $\epsilon=0$, is i.i.d.\ Laplacian on $\vec{z}_d\!=\!\vec{\Psi}_d\vec{x}\in\Real^{L_d}$ with deterministic scale parameter $\lambda_d>0$. 
\end{enumerate}
\end{theorem}
\begin{proof} 
See Sections~\ref{sec:score_logsum}~to~\ref{sec:score_em} below.
\end{proof}

Importantly, the MM interpretation implies convergence (in the sense of an asymptotic stationary point condition) when $\epsilon>0$, as detailed in \secref{score_conv}.

\subsection{Log-Sum MM Interpretation of \score} \label{sec:score_logsum}

Consider the optimization problem
\begin{align}
\arg\min_{\vec{x}} 
\gamma\norm{\vec{y}-\vec{\Phi x}}_2^2 +
\Rlog^D(\vec{x};\epsilon)
\label{eq:score_logsum}
\end{align}
with $\Rlog^D$ from \eqref{Rlog}.
Inspired by \cite[\S 2.3]{Candes:JFA:08}, we write \eqref{score_logsum} as
\begin{align}
&\arg\min_{\vec{x},\vec{u}} \gamma\norm{\vec{y}-\vec{\Phi x}}_2 
+ \sum_{d=1}^D L_d
\log\bigg(\epsilon + \sum_{l=1}^{L_d}u_{d,l}\bigg) 
\nonumber\\&\quad
\st |\vec{\psi}_{d,l}\tran\vec{x}|\leq u_{d,l} ~\forall d,l 
\label{eq:score_logsum2} ,
\end{align}
where $\vec{\psi}_{d,l}\tran$ is the $l$th row of $\vec{\Psi}_d$.
Problem \eqref{score_logsum2} is of the form
\begin{align}
\arg\min_{\vec{v}} g(\vec{v}) 
\st \vec{v}\in\mc{C},
\label{eq:convex}
\end{align}
where 
$\vec{v}=[\vec{u}\tran,\vec{x}\tran]\tran$,
$\mc{C}$ is a convex set, 
\begin{align} 
g(\vec{v})
= \gamma \big\| \vec{y}-[\vec{0}~\vec{\Phi}]\vec{v}\big\|_2^2 
+ \sum_{d=1}^D L_d \log\bigg(\epsilon +\sum_{k\in\mc{K}_d}v_{k}\bigg) 
\label{eq:score_g}
\end{align} 
is a non-convex penalty, and
the set $\mc{K}_d\defn \{k: \sum_{i=1}^{d-1}L_i < k \leq \sum_{i=1}^d L_i\}$ contains the indices $k$ such that $v_k\in\{u_{d,l}\}_{l=1}^{L_d}$.

Since $g(\vec{v})$ is the sum of convex and concave terms, i.e., a ``difference of convex'' (DC) functions, \eqref{convex} can be recognized as a DC program \cite{Horst:JOTA:99}.
Majorization-minimization (MM) \cite{Hunter:AS:04,Mairal:ICML:13} is a popular method to attack non-convex problems of this form.
In particular, MM iterates the following two steps: (i) construct a surrogate $g(\vec{v};\vec{v}\of{t})$ that majorizes $g(\vec{v})$ at $\vec{v}\of{t}$, and (ii) update $\vec{v}\of{t+1}=\arg\min_{\vec{v}\in\mc{C}} g(\vec{v};\vec{v}\of{t})$.
By ``majorize,'' we mean that $g(\vec{v};\vec{v}\of{t}) \geq g(\vec{v})$ for all $\vec{v}$ with equality when $\vec{v}=\vec{v}\of{t}$.

Due to the DC form of $g(\vec{v})$ in \eqref{score_g}, a majorizing surrogate can be constructed by linearizing the concave term about its tangent at $\vec{v}\of{t}$.
In particular, say $g(\vec{v})=g_1(\vec{v})+g_2(\vec{v})$, where $g_1$ is the convex (quadratic) term and $g_2$ is the concave (log-sum) term, 
and say $\nabla g_2$ is the gradient of $g_2$ w.r.t.\ $\vec{v}$. 
Then
\begin{align}
g(\vec{v};\vec{v}\of{t})
\defn g_1(\vec{v}) + g_2(\vec{v}\of{t}) + \nabla g_2(\vec{v}\of{t})\tran [\vec{v}-\vec{v}\of{t}]
\end{align}
majorizes $g(\vec{v})$ at $\vec{v}\of{t}$, and so the MM iterations become
\begin{align}
\vec{v}\of{t+1}
&= \arg\min_{\vec{v}\in\mc{C}} g_1(\vec{v}) + \nabla g_2(\vec{v}\of{t})\tran\vec{v} 
\label{eq:mm}
\end{align}
after neglecting the $\vec{v}$-invariant terms.

Examining the log-sum term in \eqref{score_g}, we see that 
\begin{align}
[\nabla g_2(\vec{v}\of{t})]_k = 
\begin{cases}
\displaystyle
\frac{L_{d(k)}}{\epsilon + \sum_{i\in\mc{K}_{d(k)}} v_{i}\of{t}} & \text{if $d(k)\neq 0$} \\
0 & \text{else},
\end{cases}
\label{eq:score_grad}
\end{align}
where $d(k)$ is the index $d\in\{1,...,D\}$ of the set $\mc{K}_d$ containing $k$, or $0$ if no such set exists.
Thus MM prescribes
\begin{align}
\vec{v}\of{t+1}
&= \arg\min_{\vec{v}\in\mc{C}} 
\gamma \big\| \vec{y}-[\vec{0}~\vec{\Phi}]\vec{v}\big\|_2^2 + 
\sum_{d=1}^D \sum_{k\in\mc{K}_d} \frac{L_d v_k}{\epsilon+\sum_{i\in\mc{K}_d} v_{i}\of{t}} ,
\end{align}
or equivalently
\begin{align}
\vec{x}\of{t+1}
&= \arg\min_{\vec{x}} 
\gamma \norm{\vec{y}-\vec{\Phi x}}_2^2 + 
\sum_{d=1}^D \frac{L_d \sum_{l=1}^{L_d} |\vec{\psi}_{d,l}\tran\vec{x}|}{\epsilon+\sum_{l=1}^{L_d} |\vec{\psi}_{d,l}\tran\vec{x}\of{t}|} \\
&= \arg\min_{\vec{x}} 
\gamma \norm{\vec{y}-\vec{\Phi x}}_2^2 + 
\sum_{d=1}^D \lambda_d\of{t+1} \norm{\vec{\Psi}_d\vec{x}}_1 
\end{align}
for
\begin{align}
\lambda_d\of{t+1}
&=\frac{L_d}{\epsilon+\norm{\vec{\Psi}_d\vec{x}\of{t}}_1} ,
\end{align}
which coincides with \algref{score}. 
This establishes \partref{score_logsum} of \thmref{score}.

\subsection{Convergence of \score} \label{sec:score_conv}

The paper \cite{Mairal:ICML:13} studies the convergence of MM, and includes a special discussion of the application of MM to DC programming. 
In the language of our \secref{score_logsum},
\cite{Mairal:ICML:13} establishes that, when $g_2$ is differentiable with a Lipschitz continuous gradient, the MM sequence $\{\vec{v}\of{t}\}_{t \geq 1}$ satisfies an asymptotic stationary point (ASP) condition. 
Although this falls short of establishing convergence to a local minimum (which is difficult for generic non-convex problems), the ASP condition is based on a classical necessary condition for a local minimum. 
In particular, using $\nabla g(\vec{v};\vec{d})$ to denote the directional derivative of $g$ at $\vec{v}$ in the direction $\vec{d}$, it is known \cite{Borwein:Book:06} that $\vec{v}_\star$ locally minimizes $g$ over $\mc{C}$ only if $\nabla g(\vec{v}_\star;\vec{v}-\vec{v}_\star)\geq 0$ for all $\vec{v}\in\mc{C}$.
Thus, in \cite{Mairal:ICML:13}, it is said that $\{\vec{v}\of{t}\}_{t\geq 1}$ satisfies an ASC condition if
\begin{align}
\liminf_{t\rightarrow +\infty} \inf_{\vec{v}\in\mc{C}}
\frac{\nabla g(\vec{v}\of{t}; \vec{v}-\vec{v}\of{t})}{\|\vec{v}-\vec{v}\of{t}\|_2} \geq 0 .
\label{eq:ASP}
\end{align}

In our case, $g_2$ from \eqref{score_g} is indeed differentiable, with gradient $\nabla g_2$ given by \eqref{score_grad}. 
Moreover, \appref{score_lipschitz} shows that this gradient is Lipschitz continuous when $\epsilon>0$. 
Thus, the sequence of estimates produced by \algref{score} satisfies the ASP condition \eqref{ASP}.

\subsection{Approximate \texorpdfstring{$\ell_{1,0}$}{L10} Interpretation of \score} \label{sec:score_ell0}

In the limit of $\epsilon\rightarrow 0$, the log-sum minimization 
\begin{align}
\arg\min_{\vec{x}}
\gamma \norm{\vec{y}-\vec{\Phi x}}_2^2 + 
\sum_{n=1}^N \log(\epsilon+|x_n|)
\label{eq:logsum}
\end{align}
for $\gamma>0$ is known \cite{Wipf:JSTSP:10} to be equivalent to $\ell_0$ minimization 
\begin{align}
\arg\min_{\vec{x}} 
\gamma' \norm{\vec{y}-\vec{\Phi x}}_2^2 + 
\|\vec{x}\|_0
\label{eq:l0}
\end{align}
for some $\gamma'>0$.  (See \appref{l0} for a proof.)
This equivalence can be seen intuitively as follows. 
As $\epsilon\rightarrow 0$, the contribution to the regularization term $\sum_{n=1}^N \log(\epsilon+|x_n|)$ from each non-zero $x_n$ remains finite, while that from each zero-valued $x_n$ approaches $-\infty$.
Since we are interested in minimizing the regularization term, we get a huge reward for each zero-valued $x_n$, or---equivalently---a huge penalty for each non-zero $x_n$.

To arrive at an $\ell_0$ interpretation of the \score algorithm, we consider the corresponding optimization problem \eqref{score_logsum} in the limit that $\epsilon\rightarrow 0$. 
There we see that the regularization term $\Rlog^D(\vec{x};0)$ from \eqref{Rlog} yields $L_d$ huge rewards when $\|\vec{\Psi}_d\vec{x}\|_1\!=\!0$, or equivalently $L_d$ huge penalties when $\|\vec{\Psi}_d\vec{x}\|_1\neq 0$, for each $d\in\{1,\dots,D\}$.
Thus, we can interpret \score as attempting to solve the optimization problem \eqref{score_ell0}, which is a weighted version of the ``$\ell_{p,q}$ mixed norm'' problem from \cite{Kowalski:ACHA:09} for $p\!=\!1$ and $q\rightarrow0$.
This establishes \partref{score_ell0} of \thmref{score}.

\subsection{Bayesian MAP Interpretation of \score} \label{sec:score_bayes}

The MAP estimate \cite{Poor:Book:94} of $\vec{x}$ from $\vec{y}$ is
\begin{align}
\vec{x}\map
&\defn \argmax_{\vec{x}} p(\vec{x}|\vec{y})
= \argmin_{\vec{x}} \big\{ -\log p(\vec{x}|\vec{y}) \big\} 
\label{eq:map_unconstrained0} \\
&= \argmin_{\vec{x}} \big\{ -\log p(\vec{x}) - \log p(\vec{y}|\vec{x}) \big\} 
\label{eq:map_unconstrained} \\
&= \argmin_{\vec{x}} \bigg\{ -\log p(\vec{x}) + \gamma\|\vec{y}-\vec{\Phi x}\|_2^2 \bigg\} 
\label{eq:map} ,
\end{align}
where \eqref{map_unconstrained0} used the monotonicity of $\log$,
\eqref{map_unconstrained} used Bayes rule, and  
\eqref{map} used the AWGN likelihood.
Note that, for real-valued AWGN with $\sigma^2$ variance, $\gamma=\frac{1}{2\sigma^2}$, while for circular complex-valued AWGN with $\sigma^2$ variance, $\gamma=\frac{1}{\sigma^2}$.

Next, we derive the $-\log p(\vec{x})$ term in \eqref{map} that results from the hierarchical prior \eqref{score_bayes_prior_x}-\eqref{score_bayes_prior_lam}.
Recall that, with shape parameter $\kappa$ and scale parameter $\theta$, the Gamma pdf \cite{Berger:Book:85} is 
$\Gamma(\lambda_d;\kappa,\theta)
=1_{\lambda_d>0}\lambda_d^{\kappa-1} \theta^{-\kappa} \exp(-\lambda_d/\theta) / \Gamma(\kappa)$, where $\Gamma(\kappa)$ is the Gamma function.
Since $\Gamma(\lambda_d;\kappa,\theta)\propto 1_{\lambda_d>0}\lambda_d^{\kappa-1}\exp(-\lambda_d/\theta)$, we note that $\Gamma(\lambda_d;0,\infty)\propto 1_{\lambda_d>0}/\lambda_d$, which is Jeffrey's non-informative hyperprior \cite{Berger:Book:85,Figueiredo:TIP:01} for the Laplace scale parameter $\lambda_d$.
Then, according to \eqref{score_bayes_prior_x}-\eqref{score_bayes_prior_lam}, the prior equals
\begin{align}
\lefteqn{
p(\vec{x})
= \int_{\Real^D} p(\vec{x}|\vec{\lambda}) p(\vec{\lambda}) \deriv\vec{\lambda}
}\\
&\propto \prod_{d=1}^D \int_0^\infty \bigg(\frac{\lambda_d}{2}\bigg)^{L_d} 
\exp({-\lambda_d}\|\vec{\Psi}_d\vec{x}\|_1) 
\frac{\exp(-\lambda_d\epsilon)}{\lambda_d} \deriv\lambda_d\\
&= \prod_{d=1}^D \frac{(L_d-1)!}{\big(2(\|\vec{\Psi}_d\vec{x}\|_1+\epsilon)\big)^{L_d}}
\end{align}
which implies that
\begin{align}
-\log p(\vec{x})
&= \const + \sum_{d=1}^D L_d \log \big(\|\vec{\Psi}_d\vec{x}\|_1+\epsilon\big) .
\label{eq:score_bayes_logprior_x} 
\end{align}

Equations \eqref{map}, \eqref{score_bayes_logprior_x}, and \eqref{Rlog} imply
\begin{align}
\vec{x}\map 
&= \argmin_{\vec{x}} 
\gamma\|\vec{y}-\vec{\Phi x}\|_2^2 + 
\Rlog^D(\vec{x};0) .
\end{align}
Finally, applying the MM algorithm to this optimization problem (as detailed in \secref{score_logsum}), we arrive at 
\algref{score}.
We note that \cite{Oliveira:SP:09} proposed to use Gamma and Jeffrey's hyperpriors with MM for total-variation image deblurring, although their algorithm is not of the IRW-L1 form.
This establishes \partref{score_bayes} of \thmref{score}.

\subsection{Variational EM Interpretation of \score} \label{sec:score_em}

The variational expectation-maximization (VEM) algorithm \cite{Neal:Jordan:98,Bishop:Book:07} is an iterative approach to maximum-likelihood (ML) estimation that generalizes the EM algorithm from \cite{Dempster:JRSS:77}.
We now provide a brief review of the VEM algorithm and describe how it can be applied to estimate $\vec{\lambda}$ in \eqref{score_em_prior_x}. 

First, note that the log-likelihood can be written as
\begin{align}
\lefteqn{ \log p(\vec{y};\vec{\lambda})
= \int q(\vec{x}) \log p(\vec{y};\vec{\lambda}) \deriv\vec{x} } \\
&= \int q(\vec{x}) \log \left[
        \frac{p(\vec{x},\vec{y};\vec{\lambda})}{q(\vec{x})}
        \frac{q(\vec{x})}{p(\vec{x}|\vec{y};\vec{\lambda})} \right]
        \deriv\vec{x} \\
&= \underbrace{
   \int q(\vec{x}) \log \frac{p(\vec{x},\vec{y};\vec{\lambda})}{q(\vec{x})}
        \deriv\vec{x} 
   }_{\displaystyle \defn F\big(q(\vec{x});\vec{\lambda}\big)}
   + \underbrace{
   \int q(\vec{x}) \log \frac{q(\vec{x})}{p(\vec{x}|\vec{y};\vec{\lambda})}
        \deriv\vec{x} 
   }_{\displaystyle \defn 
        \Dkl\big( q(\vec{x}) \big\| p(\vec{x}|\vec{y};\vec{\lambda}) \big)
   } 
   \label{eq:free} ,
\end{align}
for an arbitrary pdf $q(\vec{x})$, where $\Dkl(q\|p)$ denotes the KL divergence of $p$ from $q$.
Because $\Dkl(q\|p)\geq 0$ for any $q$ and $p$, we see that $F(q(\vec{x});\vec{\lambda})$ is a lower bound on $\log p(\vec{y};\vec{\lambda})$.
The EM algorithm performs ML estimation by iterating
\begin{align}
q\of{t}(\vec{x})
&= \arg\min_{q} \Dkl\big( q(\vec{x}) \big\| p(\vec{x}|\vec{y};\vec{\lambda}\of{t}) \big)
\label{eq:Estep} \\
\vec{\lambda}\of{t+1}
&= \arg\max_{\vec{\lambda}} F(q\of{t}(\vec{x});\vec{\lambda}) 
\label{eq:Mstep} ,
\end{align}
where the ``E'' step \eqref{Estep} tightens the lower bound and the ``M'' step \eqref{Mstep} maximizes the lower bound. 

The EM algorithm places no constraints on $q(\vec{x})$, in which case the solution to \eqref{Estep} is simply $q\of{t}(\vec{x})=p(\vec{x}|\vec{y};\vec{\lambda}\of{t})$, i.e., the posterior pdf of $\vec{x}$ under $\vec{\lambda}=\vec{\lambda}\of{t}$.
In many applications, however, this posterior is too difficult to compute and/or use in \eqref{Mstep}.
To circumvent this problem, the VEM algorithm constrains $q(\vec{x})$ to some family of distributions $\mc{Q}$ that makes \eqref{Estep}-\eqref{Mstep} tractable.

For our application of the VEM algorithm, we constrain to distributions of the form
\begin{align}
q(\vec{x})
&\propto \lim_{T\rightarrow 0} \exp \Big( \tfrac{1}{T} \log p(\vec{x}|\vec{y};\vec{\lambda}) \Big) ,
\end{align}
which has the effect of concentrating the mass in $q(\vec{x})$ at its mode.
Plugging this $q(\vec{x})$ and $p(\vec{x},\vec{y};\vec{\lambda})=p(\vec{y}|\vec{x})p(\vec{x};\vec{\lambda})$ into \eqref{free}, we see that the M step \eqref{Mstep} reduces to
\begin{align}
\vec{\lambda}\of{t+1}
&= \arg\max_{\vec{\lambda}} \log p(\vec{x};\vec{\lambda})\barst{\vec{x}=\vec{x}\map\of{t}}
\label{eq:Mstep_score} \\
\text{for~~}
\vec{x}\map\of{t}
&\defn \arg\max_{\vec{x}} p(\vec{x}|\vec{y};\vec{\lambda}\of{t})
\label{eq:Estep_score} ,
\end{align}
where \eqref{Estep_score} be interpreted as the E step.
For the particular $p(\vec{x};\vec{\lambda})$ in \eqref{score_em_prior_x}, we have  that
\begin{align}
\log p(\vec{x};\vec{\lambda}) 
&= \const + \sum_{d=1}^D \big[ L_d \log(\lambda_d)-\lambda_d (\|\vec{\Psi}_d\vec{x}\|_1 +\epsilon)\big] ,
\label{eq:logpx_score}
\end{align}
and by zeroing the gradient w.r.t.\ $\vec{\lambda}$, we find that \eqref{Mstep_score} becomes
\begin{align}
\lambda_d\of{t+1}
= \frac{L_d}{\big\|\vec{\Psi}_d\vec{x}\of{t}\map\big\|_1+\epsilon},~~d=1,\dots,D.
\label{eq:Mstep_score2}
\end{align}
Meanwhile, from \eqref{map} and \eqref{score_em_prior_x}, we find that \eqref{Estep_score} becomes
\begin{align}
\vec{x}\of{t}\map
&= \argmin_{\vec{x}} 
\gamma\|\vec{y}-\vec{\Phi x}\|_2^2 + 
\sum_{d=1}^D \lambda_d\of{t} \|\vec{\Psi}_d\vec{x}\|_1 
\label{eq:Estep_score2} .
\end{align}

In conclusion, our VEM algorithm iterates the steps \eqref{Mstep_score2}-\eqref{Estep_score2}, which match the steps in \algref{score}. 
This establishes \partref{score_em} of \thmref{score}.

\subsection{\score for Complex-Valued \texorpdfstring{$\vec{\Psi}_d \vec{x}$}{Analysis Outputs}} \label{sec:score_complex}

In \thmref{score} and Sections~\ref{sec:score_logsum}-\ref{sec:score_em}, real-valued analysis outputs $\vec{\Psi}_d \vec{x}$ were assumed for ease of explanation.
We now extend the previous results to the case of complex-valued $\vec{\Psi}_d \vec{x}$.
For this, we focus on the VEM interpretation (recall \partref{score_em} of \thmref{score}), noting that a similar justification can be made based on the Bayesian MAP interpretation.
In particular, we 
assume an AWGN likelihood and a complex-valued extension of the prior \eqref{score_em_prior_x}:
\begin{align}
p(\vec{x};\vec{\lambda}) 
&\propto \prod_{d=1}^D \bigg(\frac{\lambda_d}{2\pi}\bigg)^{2L_d} \exp\big({-\lambda_d} (\|\vec{\Psi}_d\vec{x}\|_1+\epsilon)\big) ,
\label{eq:score_em_prior_x_cmplx} 
\end{align}
which, when $\epsilon=0$, is i.i.d.\ complex-valued Laplacian on $\vec{z}_d\!=\!\vec{\Psi}_d\vec{x}\in\Complex^{L_d}$ with deterministic scale parameter $\lambda_d>0$. 
To show this, we follow the steps in \secref{score_em} up to the log-prior in \eqref{logpx_score}, which now becomes
\begin{align}
\log p(\vec{x};\vec{\lambda}) 
&= \const + \sum_{d=1}^D \big[ 2 L_d \log(\lambda_d)-\lambda_d (\|\vec{\Psi}_d\vec{x}\|_1+\epsilon) \big] . \quad
\label{eq:logpx_score_cmplx}
\end{align}
Zeroing the gradient w.r.t.\ $\vec{\lambda}$, we find that the VEM update in \eqref{Mstep_score} becomes
\begin{align}
\lambda_d\of{t+1}
= \frac{2L_d}{\big\|\vec{\Psi}_d\vec{x}\of{t}\map\big\|_1+\epsilon},~~d=1,\dots,D,
\label{eq:Mstep_score2_cmplx}
\end{align}
which is twice as large as the real-valued case in \eqref{Mstep_score2}.

\subsection{New Interpretations of the \IRW Algorithm} \label{sec:irw}

The proposed \score algorithm is related to the analysis-CS formulation 
of the well-known \IRW algorithm \cite{Figueiredo:TIP:07}. 
For clarity, and for later use in \secref{scoreirw}, we summarize this latter algorithm in \algref{irw}, and note that the synthesis-CS formulation follows from the special case that $\vec{\Psi}=\vec{I}$.

\begin{algorithm}
  \caption{The \IRW Algorithm}
  \label{alg:irw}
  \begin{algorithmic}[1]
    \State
    \textsf{input:~~} $\vec{\Psi}=[\vec{\psi}_1,\dots,\vec{\psi}_L]\tran$, $\vec{\Phi}$, $\vec{y}$, 
    $\gamma\geq 0$,
    $\epsilon\geq 0$\\
    \textsf{initialization:~~} $\vec{W}\of{1} = \vec{I}$
    \State
    \textsf{for~} $t=1,2,3,\dots$\\
    \quad $\vec{x}\of{t} \gets \displaystyle \argmin_{\vec{x}} 
        \gamma \|\vec{y} - \vec{\Phi x}\|^2_2 +
        \|\vec{W}\of{t}\vec{\Psi}\vec{x}\|_{1}$ 
    \\
    \quad $\vec{W}\of{t+1} \gets \diag\left\{ \displaystyle
        \frac{1}{\epsilon + |\vec{\psi}_1\tran \vec{x}\of{t}|},
        \cdots, 
        \frac{1}{\epsilon + |\vec{\psi}_L\tran \vec{x}\of{t}|}
        \right\}$\\
    \textsf{end}\\
     \textsf{output:~~}$\vec{x}\of{t}$
  \end{algorithmic}
\end{algorithm}

Comparing \algref{irw} to \algref{score}, we see that \IRW coincides with real-valued \score in the case that every sub-dictionary $\vec{\Psi}_d$ has dimension one, i.e., $C_d\!=\!1\!=\!L_d\!~\forall d$ and $D\!=\!L$,
where $L\defn\sum_{d=1}^D L_d$ denotes the total number of analysis coefficients.
Thus, the \score interpretations from \thmref{score} can be directly translated to \IRW as follows.

\begin{corollary}[\IRW]\label{cor:irw}
The \IRW algorithm from \algref{irw} has the following interpretations:
\begin{enumerate}
\item \label{part:irw_logsum}
MM applied to \eqref{opt} under the log-sum penalty 
\begin{align}
\Rlog^L(\vec{x};\epsilon) 
&= \sum_{l=1}^L \log( \epsilon +|\vec{\psi}_l\tran \vec{x}| ) ,
\label{eq:RlogL}
\end{align}
recalling the definition of $\Rlog^L$ from \eqref{Rlog},

\item \label{part:irw_ell0}
as $\epsilon\rightarrow 0$,
MM applied to \eqref{opt} under the $\ell_{0}$ penalty
\begin{align}
R_{0}^L(\vec{x})
&\defn \sum_{l=1}^L 1_{|\vec{\psi}_l\tran\vec{x}|>0}
\label{eq:irw_ell0} ,
\end{align}

\item \label{part:irw_bayes}
MM applied to Bayesian MAP estimation under an AWGN likelihood and the hierarchical prior
\begin{align}
p(\vec{x}|\vec{\lambda}) 
&= \prod_{l=1}^L \frac{\lambda_l}{2} \exp\big({-\lambda_l} |\vec{\psi}_l\tran\vec{x}|\big)
\label{eq:irw_bayes_prior_x} \\
\vec{\lambda} &\sim \text{i.i.d.~}\Gamma(0,\epsilon^{-1})
\label{eq:irw_bayes_prior_lam} 
\end{align}
where $z_l\!=\!\vec{\psi}_l\tran\vec{x}$ is Laplacian given $\lambda_l$, and $\lambda_l$ is Gamma distributed with scale parameter $\epsilon^{-1}$ and shape parameter zero, which becomes
Jeffrey's non-informative hyperprior $p(\lambda_l)\propto 1_{\lambda_l>0}/\lambda_l$ when $\epsilon=0$.  

\item \label{part:irw_em}
variational EM under an AWGN likelihood and the prior
\begin{align}
p(\vec{x};\vec{\lambda}) 
&\propto \prod_{l=1}^L \frac{\lambda_l}{2}\exp\big({-\lambda_l} (|\vec{\psi}_l\tran\vec{x}|+\epsilon)\big) .
\label{eq:irw_em_prior_x} 
\end{align}
which, when $\epsilon=0$, is independent Laplacian on $\vec{z}\!=\!\vec{\Psi x}\in\Real^L$ under the positive deterministic scale parameters in $\vec{\lambda}$.
\end{enumerate}
\end{corollary}

While \partref{irw_logsum} and \partref{irw_ell0} of \corref{irw} were established for the $\ell_2$-constrained synthesis-CS formulation of \IRW in \cite{Candes:JFA:08}, we believe that \partref{irw_bayes} and \partref{irw_em} are novel interpretations of \IRW.

\section{The \ascoreirw algorithm} \label{sec:scoreirw}

We now propose the \scoreirw algorithm, which is summarized in \algref{scoreirw}.
\scoreirw can be thought of as a hybrid of the \score and \IRW approaches from Algorithms~\ref{alg:score}~and~\ref{alg:irw}, respectively.
Like with \score, the \scoreirw algorithm uses sub-dictionary dependent weights $\lambda_d$ that are updated at each iteration $t$ using a sparsity metric on $\vec{\Psi}_d\vec{x}\of{t}$.
But, like with \IRW, the \scoreirw algorithm also uses diagonal weight matrices $\vec{W}_d\of{t}$ that are updated at each iteration.
As with both \score and \IRW, the computational burden of \scoreirw is dominated by the L2+L1 minimization problem in \lineref{scoreirw_x} of \algref{scoreirw}, which is readily solved by existing techniques like MFISTA. 
\begin{algorithm}
  \caption{The Real-Valued \scoreirw Algorithm}
  \label{alg:scoreirw}
  \begin{algorithmic}[1]
    \State
    \textsf{input:~~} $\{\vec{\Psi}_d\}_{d=1}^D$, $\vec{\Phi}$, $\vec{y}$, 
    $\gamma>0$,
    $\epsilon_d>0~\forall d$, $\varepsilon\geq 0$, \\
    \textsf{initialization:~~} $\lambda_d\of{1}=1~\forall d$,
    $\vec{W}_d\of{1} = \vec{I}~\forall d$
    \State
    \textsf{for~} $t=1,2,3,\dots$\\
    \label{line:scoreirw_x}
    \quad $\vec{x}\of{t} \gets 
        \begin{array}[t]{@{}l}
        \displaystyle \argmin_{\vec{x}} 
        \gamma\|\vec{y}-\vec{\Phi x}\|_2^2 +
        \sum_{d=1}^D \lambda_d\of{t}\|\vec{W}_d\of{t}\vec{\Psi}_d\vec{x} \|_1 \\
        \end{array}$\\
    \quad 
    \label{line:scoreirw_lam} 
    $\lambda_d\of{t+1}
    \gets \begin{array}[t]{@{}l}
        \displaystyle \left[ \frac{1}{L_d}\sum_{l=1}^{L_d}
        \log\bigg(1 +\varepsilon
        +\frac{|\vec{\psi}_{d,l}\tran\vec{x}\of{t}|}{\epsilon_d}
        \bigg) \right]^{-1} + 1, \\
        \forall d=1,...,D 
        \end{array} $ \\
    \label{line:scoreirw_wgt} 
    \quad 
    $\vec{W}_d\of{t+1} \gets 
        \begin{array}[t]{@{}l}
        \displaystyle 
        \diag\bigg\{
        \frac{1}{\epsilon_d(1+\varepsilon)+|\vec{\psi}_{d,1}\tran\vec{x}\of{t}|}, 
        \cdots,
        \\\hspace{10mm} \displaystyle 
        \frac{1}{\epsilon_d(1+\varepsilon)+|\vec{\psi}_{d,L_d}\tran\vec{x}\of{t}|}
        \bigg\},
        ~\forall d
        \end{array}
        $\\
    \textsf{end}\\     
    \textsf{output:~~}$\vec{x}\of{t}$  
  \end{algorithmic}
\end{algorithm}

THE \scoreirw algorithm can be interpreted in various ways, as we detail below.
For clarity, we first consider fixed regularization parameters $\vec{\epsilon}\defn [\epsilon_1,\dots,\epsilon_D]\tran$ and later, in \secref{scoreirw_epsilon}, we describe how they can be adapted at each iteration, leading to the \ascoreirw algorithm.
Also, to simplify the development, we first consider the real-valued case and discuss the complex-valued case later, in \secref{scoreirw_complex}.

\begin{theorem}[\scoreirw] \label{thm:scoreirw}
The real-valued \scoreirw algorithm in \algref{scoreirw} has the following interpretations:
\begin{enumerate}
\item \label{part:scoreirw_logsum}
MM applied to \eqref{opt} under the 
log-sum-log penalty 
\begin{align}
\Rlsl(\vec{x};\vec{\epsilon},\varepsilon)
&\defn \sum_{d=1}^D \sum_{l=1}^{L_d} \log\bigg[
  \big(\epsilon_d(1+\varepsilon) + |\vec{\psi}_{d,l}\tran\vec{x}|\big)
  \nonumber\\&\quad \times
  \sum_{i=1}^{L_d} \log\bigg(1+\varepsilon
        +\frac{|\vec{\psi}_{d,i}\tran\vec{x}|}{\epsilon_d}
        \bigg) \bigg],
  \label{eq:Rlsl} 
\end{align}

\item \label{part:scoreirw_ell0}
as $\varepsilon\rightarrow 0$ and $\epsilon_d\rightarrow 0~\forall d$, 
MM applied to \eqref{opt} under the $\ell_0+\ell_{0,0}$ penalty
\begin{align}
R_{0,00}^D(\vec{x})
&\defn \|\vec{\Psi x}\|_0 + \sum_{d=1}^D L_d \,1_{\|\vec{\Psi}_d\vec{x}\|_0>0}
\label{eq:scoreirw_ell0} ,
\end{align}

\item \label{part:scoreirw_bayes}
MM applied to Bayesian MAP estimation under an AWGN likelihood and the hierarchical prior
\begin{eqnarray}
p(\vec{x}|\vec{\lambda};\vec{\epsilon}) 
&\propto& \prod_{d=1}^D \prod_{l=1}^{L_d}\frac{\lambda_d}{2\epsilon_d} 
        \bigg(1+\varepsilon
        +\frac{|\vec{\psi}_{d,l}\tran\vec{x}|}{\epsilon_d}
        \bigg)^{-(\lambda_d+1)}
\label{eq:scoreirw_bayes_prior_x} \\
p(\vec{\lambda})
&=& \prod_{d=1}^D p(\lambda_d),
~~
p(\lambda_d)\propto \begin{cases}\frac{1}{\lambda_d}&\lambda_d>0\\0&\text{else}\end{cases},
\qquad
\label{eq:scoreirw_bayes_prior_lam} 
\end{eqnarray}
where, when $\varepsilon=0$, the variables $\vec{z}_d\!=\!\vec{\Psi}_d\vec{x}\in\Real^{L_d}$ are i.i.d.\ generalized-Pareto \cite{Cevher:NIPS:09} given $\lambda_d$, and $p(\lambda_d)$ is Jeffrey's non-informative hyperprior \cite{Berger:Book:85,Figueiredo:TIP:01} for the random shape parameter $\lambda_d$.

\item \label{part:scoreirw_em}
variational EM under an AWGN likelihood and the prior
\begin{align}
p(\vec{x};\vec{\lambda},\vec{\epsilon}) 
&\propto \prod_{d=1}^D \prod_{l=1}^{L_d}\frac{\lambda_d-1}{2\epsilon_d} 
        \bigg(1 +\varepsilon
        +\frac{|\vec{\psi}_{d,l}\tran\vec{x}|}{\epsilon_d}
        \bigg)^{-\lambda_d}
\label{eq:scoreirw_em_prior_x} 
\end{align}
where, when $\varepsilon=0$, the variables $\vec{z}_d\!=\!\vec{\Psi}_d\vec{x}\in\Real^{L_d}$ are i.i.d.\ generalized-Pareto with deterministic shape parameter $\lambda_d>1$ and scale parameter $\epsilon_d>0$. 
\end{enumerate}
\end{theorem}
\begin{proof} 
See Sections~\ref{sec:scoreirw_logsum}~to~\ref{sec:scoreirw_em} below.
\end{proof}

As with \score, the MM interpretation implies convergence (in the sense of an asymptotic stationary point condition) when $\varepsilon>0$, as detailed in \secref{scoreirw_conv}.

\subsection{Log-Sum-Log MM Interpretation of \scoreirw} \label{sec:scoreirw_logsum}

Consider the optimization problem
\begin{align}
\arg\min_{\vec{x}} 
\gamma \|\vec{y}-\vec{\Phi x}|_2^2 +
\Rlsl(\vec{x};\vec{\epsilon},\varepsilon)
\label{eq:scoreirw_logsum}
\end{align}
with $\Rlsl$ defined in \eqref{Rlsl}.
We attack this optimization problem using the MM approach detailed in \secref{score_logsum}.
The difference is that now the function $g_2$ is defined as
\begin{align} 
\lefteqn{
g_2(\vec{v})
}\nonumber\\
&=\sum_{d=1}^D \sum_{k\in\mc{K}_d} \log\bigg[
  \big(\epsilon_d(1+\varepsilon) + v_k \big)
  \sum_{i\in\mc{K}_d} \log\bigg(1 +\varepsilon
        +\frac{v_i}{\epsilon_d}
        \bigg) \bigg] \\
&= \sum_{d=1}^D \left[ 
  L_d \log 
  \sum_{i\in\mc{K}_d} \log\bigg(1 +\varepsilon
        +\frac{v_i}{\epsilon_d}
        \bigg) 
\right.\nonumber\\&\quad\left.
  + \sum_{k\in\mc{K}_d} 
  \log \big(\epsilon_d(1+\varepsilon) + v_k \big)
  \right]  ,
  \quad
  \label{eq:scoreirw_g}
\end{align} 
which has a gradient of
\begin{align}
\lefteqn{ [\nabla g_2(\vec{v}\of{t})]_k }
\label{eq:scoreirw_grad} \\
&= 
\displaystyle
\left(
\frac{ L_{d(k)} }
{\sum\limits_{i\in\mc{K}_{d(k)}}\!\!\!\log\Big(1
        +\varepsilon
        \!+\!\frac{v_i\of{t}}{\epsilon_{d(k)}}
        \Big)} 
+ 1 \right)
\frac{1}{\epsilon_{d(k)}(1+\varepsilon)+v_k\of{t}}
\end{align}
when $d(k)\neq 0$ and otherwise $[\nabla g_2(\vec{v}\of{t})]_k=0$.
Thus, recalling \eqref{mm},  MM prescribes
\begin{align}
\vec{v}\of{t+1}
&= \arg\min_{\vec{v}\in\mc{C}} \sum_{d=1}^D \sum_{k\in\mc{K}_d} 
\left(
\frac{ L_{d} }
{\sum\limits_{i\in\mc{K}_{d}}\log\Big(1+\varepsilon
        +\frac{v_i\of{t}}{\epsilon_d}
        \Big)} 
+ 1 \right)
\nonumber\\&\qquad\times
\left(\frac{v_k}{\epsilon_d(1+\varepsilon) + v_k\of{t}}\right)
+ \gamma \|\vec{y}-[\vec{0}~\vec{\Phi}]\vec{v}\|_2^2
,
\end{align}
or equivalently 
\begin{align}
\vec{x}\of{t+1}
&= \arg\min_{\vec{x}} 
\sum_{d=1}^D \sum_{l=1}^{L_d} \lambda_d\of{t+1} 
\left(\frac{|\vec{\psi}_{d,l}\tran\vec{x}|}{\epsilon_d(1+\varepsilon)+|\vec{\psi}_{d,l}\tran\vec{x}\of{t}|}\right) 
\nonumber\\&\qquad
+ \gamma \|\vec{y}-\vec{\Phi x}\|_2^2
\end{align}
for
\begin{align}
\lambda_d\of{t+1}
&= \left[\frac{1}{L_d}\sum_{l=1}^{L_d}\log\bigg(
        1 +\varepsilon 
        +\frac{|\vec{\psi}_{d,l}\tran\vec{x}\of{t}|}{\epsilon_d} 
        \bigg) \right]^{-1} + 1, 
\end{align}
which coincides with \algref{scoreirw}. 
This establishes \partref{scoreirw_logsum} of \thmref{scoreirw}.

\subsection{Convergence of \scoreirw} \label{sec:scoreirw_conv}

The convergence of \scoreirw (in the sense of an asymptotic stationary point condition) for $\varepsilon>0$ can be shown using the same procedure as in \secref{score_conv}.
To do this, we only need to verify that the gradient $\nabla g_2$ in \eqref{scoreirw_grad} is Lipschitz continuous when $\varepsilon>0$, which we do in \appref{scoreirw_lipschitz}.

\subsection{Approximate \texorpdfstring{$\ell_0+\ell_{0,0}$}{L0+L00} Interpretation of \scoreirw} \label{sec:scoreirw_ell0}

Recalling the discussion in \secref{score_ell0}, we now consider the behavior of the $\Rlsl(\vec{x};\vec{\epsilon},\varepsilon)$ regularizer in \eqref{Rlsl} as $\varepsilon\rightarrow 0$ and $\epsilon_d\rightarrow 0~\forall d$.
For this, it helps to decouple \eqref{Rlsl} into two terms:
\begin{align}
\Rlsl(\vec{x};\vec{\epsilon},\varepsilon)
&= \sum_{d=1}^D \sum_{l=1}^{L_d} \log
  \big(\epsilon_d(1+\varepsilon) + |\vec{\psi}_{d,l}\tran\vec{x}|\big)
  \label{eq:Rlsl2}
  \\&\quad
  + \sum_{d=1}^D \sum_{l=1}^{L_d} \log \bigg[
  \sum_{i=1}^{L_d} \log\bigg(1
        +\varepsilon
        +\frac{|\vec{\psi}_{d,i}\tran\vec{x}|}{\epsilon_d}
        \bigg) \bigg].
  \nonumber
\end{align}
As $\epsilon_d\rightarrow 0~\forall d$, the first term in \eqref{Rlsl2} contributes an infinite valued ``reward'' for each pair $(d,l)$ such that $|\vec{\psi}_{d,l}\tran\vec{x}|=0$, or a finite valued cost otherwise.
As for the second term, we see that 
$\lim_{\varepsilon\rightarrow 0,\epsilon_d\rightarrow 0}\sum_{i=1}^{L_d} \log\big(1+\varepsilon+|\vec{\psi}_{d,i}\tran\vec{x}|/\epsilon_d\big) = 0$ if and only if $|\vec{\psi}_{d,i}\tran\vec{x}|=0~\forall i\in\{1,\dots,L_d\}$, i.e., if and only if $\|\vec{\Psi}_d\vec{x}\|_0=0$.
And when $\|\vec{\Psi}_d\vec{x}\|_0=0$, the second term in \eqref{Rlsl2} contributes $L_d$ infinite valued rewards.
In summary, as $\varepsilon\rightarrow 0$ and $\epsilon_d\rightarrow 0~\forall d$, the first term in \eqref{Rlsl2} behaves like $\|\vec{\Psi x}\|_0$ and the second term like the weighted $\ell_{0,0}$ quasi-norm $\sum_{d=1}^D L_d 1_{\|\vec{\Psi}_d\vec{x}\|_0>0}$, as stated in \eqref{scoreirw_ell0}.
This establishes \partref{scoreirw_ell0} of \thmref{scoreirw}.

\subsection{Bayesian MAP Interpretation of \scoreirw} \label{sec:scoreirw_bayes}

To show that \scoreirw can be interpreted as Bayesian MAP estimation under the hierarchical prior \eqref{scoreirw_bayes_prior_x}-\eqref{scoreirw_bayes_prior_lam}, we first compute the prior $p(\vec{x})$.  To start,
\begin{align}
\lefteqn{
p(\vec{x})
= \int_{\Real^D} p(\vec{\lambda}) p(\vec{x}|\vec{\lambda}) \deriv\vec{\lambda}}\\
&\propto \prod_{d=1}^D \int_0^\infty \frac{1}{\lambda_d} 
\prod_{l=1}^{L_d} \frac{\lambda_d}{2\epsilon_d} 
\bigg(1 +\varepsilon
        +\frac{|\vec{\psi}_{d,l}\tran\vec{x}|}{\epsilon_d}
        \bigg)^{-(\lambda_d+1)}
\deriv\lambda_d . 
\end{align}
Writing $(1+\varepsilon+|\vec{\psi}_{d,l}\tran\vec{x}|/\epsilon_d)^{-(\lambda_d+1)} = \exp(-(\lambda_d+1) Q_{d,l})$ for $Q_{d,l} \defn \log(1+\varepsilon+|\vec{\psi}_{d,l}\tran\vec{x}|/\epsilon_d)$, we get
\begin{align}
p(\vec{x})
&\propto \prod_{d=1}^D \frac{1}{(2\epsilon_d)^{L_d}}
\int_0^\infty \lambda_d^{L_d-1} 
e^{-(\lambda_d+1)\sum_{l=1}^{L_d}Q_{d,l}}
\deriv\lambda_d .
\end{align}
Defining $Q_d\defn \sum_{l=1}^{L_d}Q_{d,l}$ and changing the variable of integration to $\tau_d \defn \lambda_d Q_d$, we find
\begin{align}
\lefteqn{
p(\vec{x})
\propto \prod_{d=1}^D \frac{e^{-Q_d}}{(2\epsilon_d Q_d)^{L_d}}
\underbrace{ \int_0^\infty \tau_d^{L_d-1} e^{-\tau_d} \deriv\tau_d }_{\displaystyle (L_d-1)! } 
}\\
&\propto \prod_{d=1}^D \Bigg[\frac{1}{\epsilon_d \sum_{i=1}^{L_d} \log(1
        +\varepsilon
        +\frac{|\vec{\psi}_{d,i}\tran\vec{x}|}{\epsilon_d}
        )}\Bigg]^{L_d}
\nonumber\\&\qquad\times
        \prod_{l=1}^{L_d} \frac{1}{1 +\varepsilon
        +\frac{|\vec{\psi}_{d,l}\tran\vec{x}|}{\epsilon_d} } \\
&= \prod_{d=1}^D \prod_{l=1}^{L_d} 
\Bigg[
\Big(\epsilon_d(1+\varepsilon)+|\vec{\psi}_{d,l}\tran\vec{x}|\Big) 
\nonumber\\&\qquad\times
\sum_{i=1}^{L_d} \log\bigg(1 +\varepsilon
        +\frac{|\vec{\psi}_{d,i}\tran\vec{x}|}{\epsilon_d} \bigg) \Bigg]^{-1} ,
\end{align}
which implies that
\begin{eqnarray}
-\log p(\vec{x})
&= \const + \Rlsl(\vec{x};\vec{\epsilon},\varepsilon)
\label{eq:scoreirw_bayes_logprior_x}
\end{eqnarray}
for $\Rlsl(\vec{x};\vec{\epsilon},\varepsilon)$ defined in \eqref{Rlsl}.

Plugging \eqref{scoreirw_bayes_logprior_x} into \eqref{map}, we see that
\begin{align}
\vec{x}\map 
&= \argmin_{\vec{x}} 
\gamma \|\vec{y}-\vec{\Phi x}\|_2^2 + 
\Rlsl(\vec{x};\vec{\epsilon},\varepsilon) ,
\end{align}
which is equivalent to the optimization problem in \eqref{scoreirw_logsum}.
We showed in \secref{scoreirw_logsum} that, by applying the MM algorithm to \eqref{scoreirw_logsum}, we arrive at \algref{scoreirw}.
This establishes \partref{scoreirw_bayes} of \thmref{scoreirw}.

\subsection{Variational EM Interpretation of \scoreirw} \label{sec:scoreirw_em}

To justify the variational EM (VEM) interpretation of \scoreirw, we closely follow the approach used for \score in \secref{score_em}.
The main difference is that now the prior takes the form of $p(\vec{x};\vec{\lambda},\vec{\epsilon})$ from \eqref{scoreirw_em_prior_x}. 
Thus, \eqref{logpx_score} becomes
\begin{align}
\lefteqn{
\log p(\vec{x};\vec{\lambda},\vec{\epsilon})
}\nonumber\\
&= 
\sum_{d=1}^D \sum_{l=1}^{L_d} \Bigg[ \log\bigg(\frac{\lambda_d-1}{\epsilon_d}\bigg) 
- \lambda_d 
\log\bigg(1 +\varepsilon
        +\frac{|\vec{\psi}_{d,l}\tran\vec{x}|}{\epsilon_d} \bigg) \Bigg] 
\nonumber\\&\quad
+ \const 
\label{eq:logpx_scoreirw}
\end{align}
and by zeroing the gradient w.r.t.\ $\vec{\lambda}$ we see that
the M step \eqref{Mstep_score2} becomes 
\begin{align}
\frac{1}{\lambda_d\of{t+1}-1}
&= \frac{1}{L_d}
\log\bigg(1 +\varepsilon
        +\frac{|\vec{\psi}_{d,l}\tran\vec{x}\map\of{t}|}{\epsilon_d}
        \bigg),
~~d=1,...,D,
\label{eq:EMirw_score_exact}
\end{align}
where again $\vec{x}\map\of{t}$ denotes the MAP estimate of $\vec{x}$ under $\vec{\lambda}=\vec{\lambda}\of{t}$.
From \eqref{map} and \eqref{scoreirw_em_prior_x}, we see that  
\begin{align}
\vec{x}\map\of{t}
&= \argmin_{\vec{x}} \sum_{d=1}^D \lambda_d\of{t} \sum_{l=1}^{L_d} \log\big( |\vec{\psi}_{d,l}\tran\vec{x}| + \epsilon_d(1+\varepsilon) \big) 
\nonumber\\&\qquad
+ \gamma\|\vec{y}-\vec{\Phi x}\|_2^2 ,
\label{eq:scoreirw_em_map}
\end{align}
which (for $\varepsilon=0$) is a $\vec{\lambda}\of{t}$-weighted version of the \IRW log-sum optimization problem (recall \partref{irw_logsum} of \corref{irw}).
To solve \eqref{scoreirw_em_map}, we apply MM.
With a small modification of the MM derivation from \secref{score_logsum}, we obtain the 2-step iteration
\begin{align}
\vec{x}\map\of{i}
&= \argmin_{\vec{x}} 
\gamma\|\vec{y}-\vec{\Phi x}\|_2^2 +
\sum_{d=1}^D\lambda_d\of{t} \|\vec{W}_d\of{i}\vec{\Psi}_d\vec{x}\|_1 
\label{eq:scoreirw_em_map_mm1} \\
\vec{W}_d\of{i+1} 
&= \diag\bigg\{
        \frac{1}{\epsilon_d(1+\varepsilon)+|\vec{\psi}_{d,1}\tran\vec{x}\of{i}|}, 
        \cdots,
\nonumber\\&\hspace{14mm}
        \frac{1}{\epsilon_d(1+\varepsilon)+|\vec{\psi}_{d,L_d}\tran\vec{x}\of{i}|}
        \bigg\}.
\label{eq:scoreirw_em_map_mm2} 
\end{align}
By using only a single MM iteration per VEM iteration, the MM index ``$i$'' can be rewritten as the VEM index ``$t$,'' in which case the VEM algorithm becomes
\begin{align}
\vec{x}\of{t}
&= \argmin_{\vec{x}} 
\gamma\|\vec{y}-\vec{\Phi x}\|_2^2 +
\sum_{d=1}^D\lambda_d\of{t} \|\vec{W}_d\of{t}\vec{\Psi}_d\vec{x}\|_1 
\label{eq:scoreirw_em1} \\
\vec{W}_d\of{t+1} 
&= \diag\bigg\{
        \frac{1}{\epsilon_d(1+\varepsilon)+|\vec{\psi}_{d,1}\tran\vec{x}\of{t}|}, 
        \dots,
\nonumber\\&\hspace{14mm}
        \frac{1}{\epsilon_d(1+\varepsilon)+|\vec{\psi}_{d,L_d}\tran\vec{x}\of{t}|}
        \bigg\}, \forall d
\label{eq:scoreirw_em2} \\
\lambda_d\of{t+1}
&= \Bigg[ \frac{1}{L_d}
\log\bigg(1 +\varepsilon
        +\frac{|\vec{\psi}_{d,l}\tran\vec{x}\of{t}|}{\epsilon_d} \bigg)
\Bigg]^{-1}+1, ~~\forall d
\label{eq:scoreirw_em3} ,
\end{align}
which matches the steps in \algref{scoreirw}. 
This establishes \partref{scoreirw_em} of \thmref{scoreirw}.

\subsection{\ascoreirw} \label{sec:scoreirw_epsilon}

Until now, we have considered the \scoreirw parameters $\vec{\epsilon}=[\epsilon_1,\dots,\epsilon_D]\tran$ to be fixed and known.
But it is not clear how to set these parameters in practice.
Thus, in this section, we describe an extension of \scoreirw that adapts the $\vec{\epsilon}$ vector at every iteration.
The resulting procedure, which we will refer to as \ascoreirw, is summarized in \algref{ascoreirw}.

\begin{algorithm}
  \caption{The \ascoreirw Algorithm}
  \label{alg:ascoreirw}
  \begin{algorithmic}[1]
    \State
    \mysf{input:~~} $\{\vec{\Psi}_d\}_{d=1}^D$, $\vec{\Phi}$, $\vec{y}$, 
    $\gamma>0$, 
    $\varepsilon\geq 0$\\
    \begin{tabular}[t]{@{}l@{}}
    \mysf{if} $\vec{\Psi x}\in\Real^{L}$, \mysf{use} $\Lambda\!=\!(1,\infty)$ \mysf{and} $\log p(\vec{x};\vec{\lambda},\vec{\epsilon})$ \mysf{from \eqref{logpx_scoreirw};}\\
    \mysf{if} $\vec{\Psi x}\in\Complex^{L}$, \mysf{use} $\Lambda\!=\!(2,\infty)$ \mysf{and} $\log p(\vec{x};\vec{\lambda},\vec{\epsilon})$ \mysf{from \eqref{logpx_scoreirw_cmplx}.}
    \end{tabular}\\
    \mysf{initialization:~~} $\lambda_d\of{1}=1~\forall d$,
    $\vec{W}_d\of{1} = \vec{I}~\forall d$
    \State
    \mysf{for~} $t=1,2,3,\dots$\\
    \label{line:ascoreirw_x}
    \quad $\vec{x}\of{t} \gets 
        \begin{array}[t]{@{}l}
        \displaystyle \argmin_{\vec{x}} 
        \gamma\|\vec{y} - \vec{\Phi x}\|_2^2 +
        \sum_{d=1}^D \lambda_d\of{t}\|\vec{W}_d\of{t}\vec{\Psi}_d\vec{x} \|_1 
        \end{array}$\\
    \quad 
    \label{line:ascoreirw_lam} 
    $(\lambda_d\of{t+1},\epsilon_d\of{t+1})
    \gets \begin{array}[t]{@{}l}
        \displaystyle 
        \arg\max_{\lambda_d\in\Lambda,\epsilon_d>0}
        \log p(\vec{x}\of{t};\vec{\lambda},\vec{\epsilon}),\\
        ~d=1,...,D 
        \end{array}$ \\
    \label{line:ascoreirw_wgt} 
    \quad 
    $\vec{W}_d\of{t+1} \gets 
        \begin{array}[t]{@{}l}
        \displaystyle 
        \diag\bigg\{
        \frac{1}{\epsilon_d\of{t+1}(1+\varepsilon)+|\vec{\psi}_{d,1}\tran\vec{x}\of{t}|}, 
        \cdots,
        \\\hspace{10mm} \displaystyle 
        \frac{1}{\epsilon_d\of{t+1}(1+\varepsilon)+|\vec{\psi}_{d,L_d}\tran\vec{x}\of{t}|}
        \bigg\},
        ~\forall d
        \end{array}$\\
    \mysf{end}\\     
    \mysf{output:~~}$\vec{x}\of{t}$  
  \end{algorithmic}
\end{algorithm}

Although there does not appear to be a closed-form solution to the joint maximization problem in \lineref{ascoreirw_lam} of \algref{ascoreirw}, it is over two real parameters and thus can be solved numerically without a significant computational burden.

\algref{ascoreirw} can be interpreted as a generalization of the VEM approach to \scoreirw that is summarized in \partref{scoreirw_em} of \thmref{scoreirw} and detailed in \secref{scoreirw_em}.
Whereas \scoreirw used VEM to estimate the $\vec{\lambda}$ parameters in the prior \eqref{scoreirw_em_prior_x} for a fixed value of $\vec{\epsilon}$, \ascoreirw uses VEM to \emph{jointly} estimate $(\vec{\lambda},\vec{\epsilon})$ in \eqref{scoreirw_em_prior_x}.
Thus, \ascoreirw can be derived by repeating the steps in \secref{scoreirw_em}, except that now the maximization of $\log p(\vec{x};\vec{\lambda},\vec{\epsilon})$ in \eqref{logpx_scoreirw} is performed jointly over $(\vec{\lambda},\vec{\epsilon})$, as reflected by \lineref{ascoreirw_lam} of \algref{ascoreirw}.

\subsection{\ascoreirw for Complex-Valued \texorpdfstring{$\vec{\Psi}_d\vec{x}$}{Analysis Outputs}} \label{sec:scoreirw_complex}

In Sections~\ref{sec:scoreirw_logsum}-\ref{sec:scoreirw_epsilon}, the analysis outputs $\vec{\Psi}_d\vec{x}$ were assumed to be real-valued.  
We now extend the previous results to the case of complex-valued $\vec{\Psi}_d\vec{x}$.
For this, we focus on the \ascoreirw algorithm, since \scoreirw follows as the special case where $\vec{\epsilon}$ is fixed at a user-supplied value.

Recalling that \ascoreirw was constructed by generalizing the VEM interpretation of \scoreirw, we reconsider this VEM interpretation for the case of complex-valued $\vec{\Psi}_d\vec{x}$.
In particular, we assume an AWGN likelihood and the following complex-valued extension of the prior \eqref{scoreirw_em_prior_x}:
\begin{align}
p(\vec{x};\vec{\lambda},\vec{\epsilon}) 
&\propto \prod_{d=1}^D \prod_{l=1}^{L_d}\frac{(\lambda_d-1)(\lambda_d-2)}{2\pi\epsilon_d^2} \bigg(1+\varepsilon+\frac{|\vec{\psi}_{d,l}\tran\vec{x}|}{\epsilon_d}\bigg)^{-\lambda_d}
\label{eq:scoreirw_em_prior_x_cmplx} 
\end{align}
which (for $\varepsilon=0$) is i.i.d.\ generalized-Pareto on $\vec{z}_d=\vec{\Psi}_d\vec{x}\in\Complex^{L_d}$ with deterministic shape parameter $\lambda_d>2$ and deterministic scale parameter $\epsilon_d>0$.
In this case, the log-prior \eqref{logpx_scoreirw} changes to
\begin{align}
\log p(\vec{x};\vec{\lambda},\vec{\epsilon})
&= \const + \sum_{d=1}^D \sum_{l=1}^{L_d} \Bigg[ \log\bigg(\frac{(\lambda_d-1)(\lambda_d-2)}{\epsilon_d^2}\bigg) 
\nonumber\\&\quad
- \lambda_d 
\log\bigg(1+\varepsilon+\frac{|\vec{\psi}_{d,l}\tran\vec{x}|}{\epsilon_d}\bigg)
\Bigg] 
\label{eq:logpx_scoreirw_cmplx}
\end{align}
which is then maximized over $(\vec{\lambda},\vec{\epsilon})$ in \lineref{ascoreirw_lam} of \algref{ascoreirw}.

\section{Numerical Results} \label{sec:num}

We now present results from a numerical study into the performance of the proposed \score and \ascoreirw methods, given as \algref{score} and \algref{ascoreirw}, respectively. 
Three experiments are discussed below, all of which focus on the problem of recovering an $N$-pixel image (or image sequence) $\vec{x}$ from $M$-sample noisy compressed measurements $\vec{y}=\vec{\Phi x}+\vec{w}$, with $M\ll N$.
In the first experiment, we recover synthetic 2D finite-difference signals;
in the second experiment, we recover the Shepp-Logan phantom and the Cameraman image; 
and in the third experiment, we recover dynamic MRI sequences, also known as ``cines.''

As discussed in \secref{related}, \score can be considered as the composite extension of the standard L1-regularized approach to analysis CS, i.e., \eqref{opt} under the non-composite L1 regularizer $R(\vec{x})=\|\vec{\Psi x}\|_1$.
Similarly, \ascoreirw can be considered as the composite extension of the standard IRW approach to the same problem.
Thus, we compare our proposed composite methods against these two non-composite methods, referring to them simply as ``L1'' and ``\IRW'' in the sequel. 

\subsection{Experimental Setup} \label{sec:setup}

For the dynamic MRI experiment, we constructed $\vec{\Phi}$ using randomly sub-sampled Fourier measurements at each time instant with a varying sampling pattern across time. 
More details are given in \secref{mriRes}.
For the other experiments, we used a ``spread spectrum'' operator \cite{Puy:JASP:12} of the form
$\vec{\Phi}=\vec{DFC}$, where $\vec{C}\in\Real^{N\times N}$ is diagonal matrix with i.i.d equiprobable $\pm 1$ entries, $\vec{F}\in\Complex^{N\times N}$ is the discrete Fourier transform (DFT), and $\vec{D}\in\Real^{M\times N}$ is a row-selection operator that selects $M$ rows of $\vec{FC}\in\Complex^{N\times N}$ uniformly at random.

In all cases, the noise $\vec{w}$ was zero-mean, white, and circular Gaussian (i.e., independent real and imaginary components of equal variance).  
Denoting the noise variance by $\sigma^2$, we define the measurement signal-to-noise ratio (SNR) as $\|\vec{y}\|_2^2/(M\sigma^2)$ and the recovery SNR of signal estimate $\hvec{x}$ as $\|\vec{x}\|_2^2 / \|\vec{x}-\hvec{x}\|^2_2$.

Note that, when $\vec{x}$ is real-valued, the measurements $\vec{y}$ will be complex-valued due to the construction of $\vec{\Phi}$. 
Thus, to allow the use of real-valued L1 solvers, we split each complex-valued element of $\vec{y}$ (and the corresponding rows of $\vec{\Phi}$ and $\vec{w}$) into real and imaginary components, resulting in a real-only model.
However, to avoid possible redundancy issues caused by the conjugate symmetry of the noiseless Fourier measurements $\vec{FCx}$, we ensured that $\vec{D}$ selected at most one sample from each complex-conjugate pair.

We used MFISTA \cite{Tan:TSP:14} to implement the L2+L1 optimization needed for all methods.
The maximum number of outer, reweighting iterations for \score and \ascoreirw was set to 16, while the maximum number of inner MFISTA iterations was set at 60, with early termination if $\|\vec{x}^{\left(t\right)}-\vec{x}^{\left(t-1\right)}\|_2 / \|\vec{x}^{\left(t\right)} \|_2 < 1\times 10^{-6}$. 
In all experiments, we used $\gamma=1/\sigma^2$ (as motivated before \eqref{score_x_con}) and $\epsilon=0=\varepsilon$.

\subsection{Synthetic 2D Finite-Difference Signals} \label{sec:sim}

Our first experiment aims to answer the following question. If we know that the sparsity of $\vec{\Psi}_1\vec{x}$ differs from the sparsity of $\vec{\Psi}_2\vec{x}$, then can we exploit this knowledge for signal recovery, even if we don't know \emph{how} the sparsities are different? This is precisely the goal of composite regularizations like \eqref{Rell}.

To investigate this question, we constructed 2D signals with finite-difference structure in both the vertical and horizontal domains.
In particular, we constructed $\vec{X}=\vec{x}_1\vec{1}\tran + \vec{1}\vec{x}_2\tran$, where both $\vec{x}_1\in\Real^{48}$ and $\vec{x}_2\in\Real^{48}$ are finite-difference signals and $\vec{1}\in\Real^{48}$ contains only ones.
The locations of the transitions in $\vec{x}_1$ and $\vec{x}_2$ were selected uniformly at random and the amplitudes of the transitions were drawn i.i.d. zero-mean Gaussian.
The total number of transitions in $\vec{x}_1$ and $\vec{x}_2$ was fixed at $28$, but the ratio of the number of transitions in $\vec{x}_1$ to the number in $\vec{x}_2$, denoted by $\alpha$, was varied from $1$ to $27$.
The case $\alpha=1$ corresponds to 
$\vec{X}$ having $14$ vertical transitions and $14$ horizontal transitions, 
while the case $\alpha=27$ corresponds to $\vec{X}$ having $27$ vertical transitions and a single horizontal transition. 
(See \figref{aniso_alpha_1} for examples.)
Finally, the signal $\vec{x}\in\Real^N$ appearing in our model \eqref{y} was created by vectorizing $\vec{X}$, yielding a total of $N=48^2=2304$ pixels.

\twoFrag{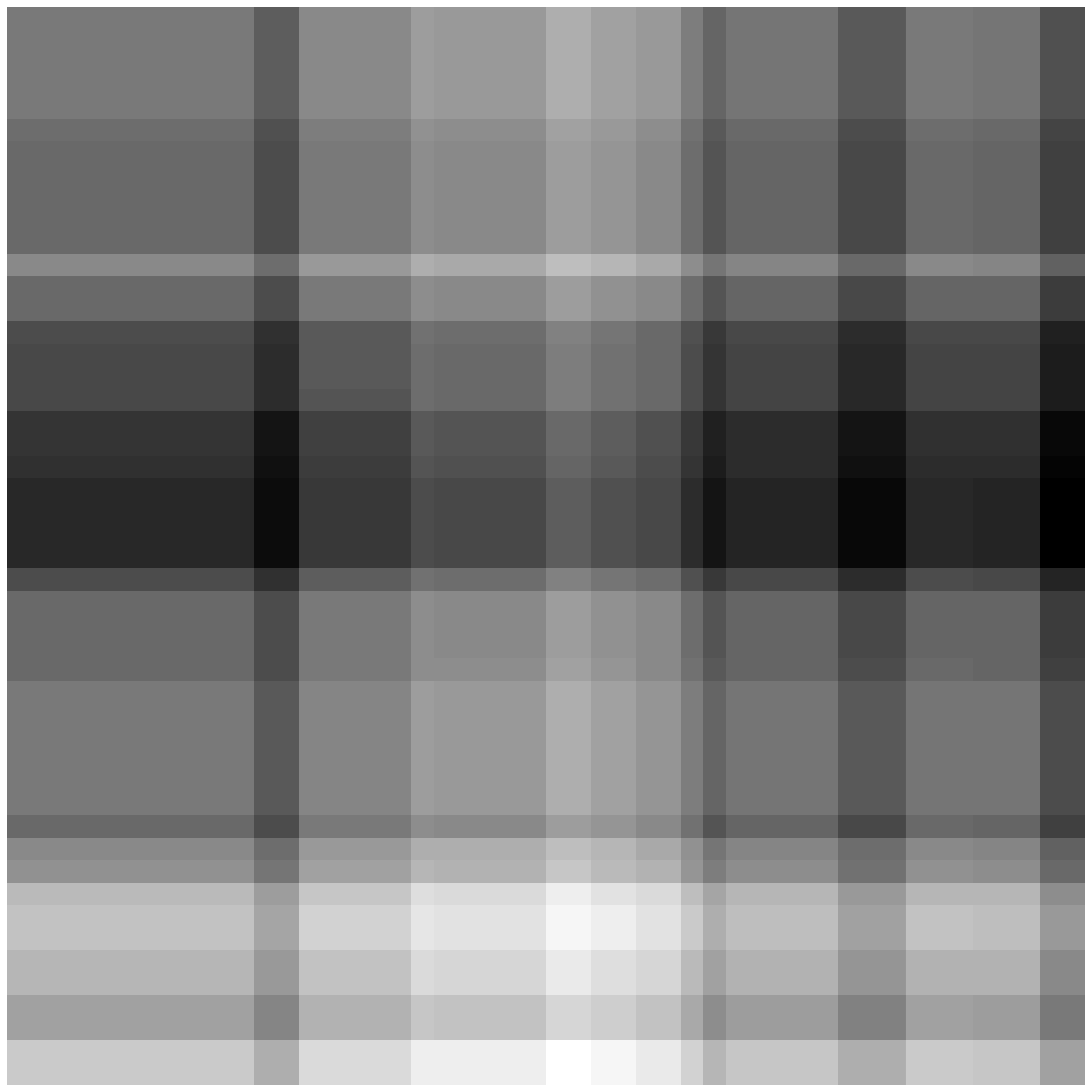}{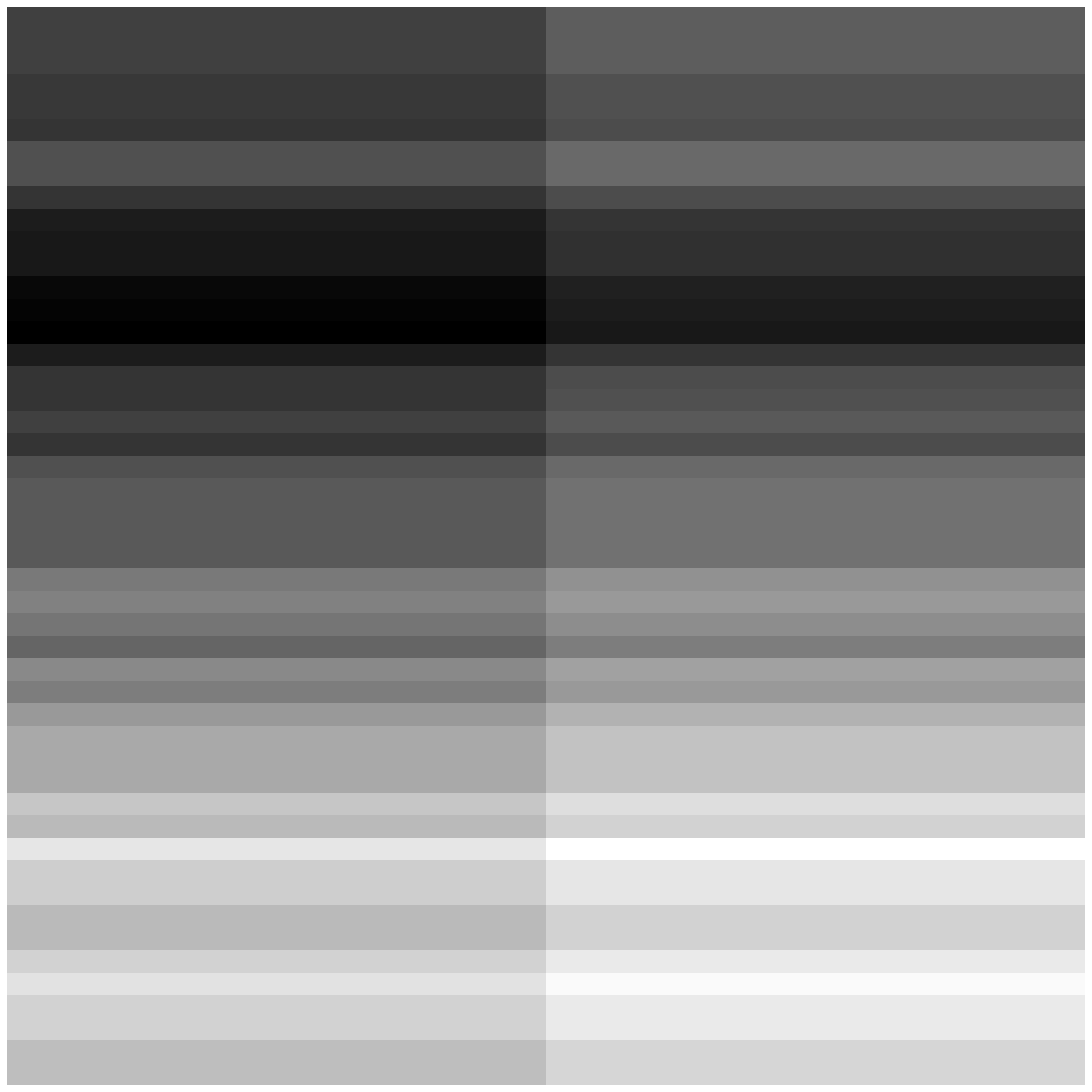} {0.4} {0.4}
        {Examples of the 2D finite-difference signal $\vec{X}$ used in the first experiment.  On the left is a realization generated under a transition ratio of $\alpha=14/14=1$, and on the right is a realization generated under $\alpha=27/1=27$.}
        {}

Given $\vec{x}$, noisy observations $\vec{y}=\vec{\Phi x}+\vec{w}$ were generated using the random ``spread spectrum'' measurement operator $\vec{\Phi}$ described earlier at a sampling ratio of $M/N=0.25$, with additive white Gaussian noise (AWGN) $\vec{w}$ scaled to achieve a measurement SNR of $40$~dB. All recovery algorithms used vertical and horizontal finite-difference operators $\vec{\Psi}_1$ and $\vec{\Psi}_2$, respectively, with $\vec{\Psi}=[\vec{\Psi}_1\tran,\vec{\Psi}_2\tran]\tran$ in the non-composite case.

\Figref{anisoRmse} shows recovery SNR versus $\alpha$ for the non-composite L1 and \IRW techniques and our proposed \score and \ascoreirw techniques.
Each SNR in the figure represents the median value from $25$ trials, each using an independent realization of the triple $(\vec{\Phi},\vec{x},\vec{w})$.
The figure shows that the recovery SNR of both L1 and \IRW is roughly invariant to the transition ratio $\alpha$, which makes sense because the overall sparsity of $\vec{\Psi x}$ is fixed at $28$ transitions by construction.
In contrast, the recovery SNRs of \score and \ascoreirw vary with $\alpha$, with higher values of $\alpha$ yielding a more structured signal and thus higher recovery SNR when this structure is properly exploited.

\putFrag{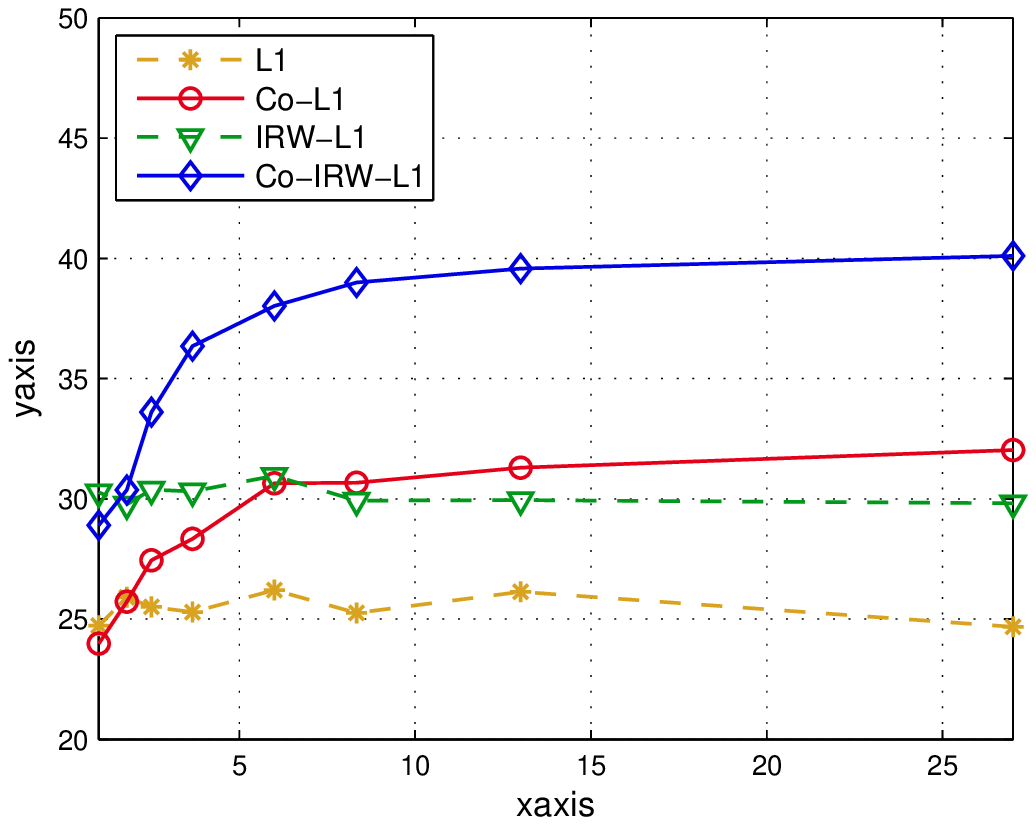}
        {Recovery SNR versus transition ratio $\alpha$ for the first experiment, which used 2D finite-difference signals, spread-spectrum measurements at $M/N=0.25$, AWGN at $40$~dB, and finite-difference operators for $\vec{\Psi}_d$.  Each recovery SNR represents the median value from $25$ independent trials.}
{\figsize}
{\newcommand{\sz}{0.9}
 \newcommand{\szz}{0.75}
 \psfrag{yaxis}[b][b][\sz]{\sf recovery SNR [dB]}
 \psfrag{xaxis}[t][t][\sz]{\sf transition ratio $\alpha$}  
 } 

\subsection{Cameraman and Shepp-Logan Recovery} \label{sec:sl}

For our second experiment, we investigate algorithm performance versus sampling ratio $M/N$ when recovering the well-known Shepp-Logan phantom and Cameraman images.
In particular, we used the $N=96\times 104$ cropped real-valued Cameraman image and the $N=96\times 96$ complex-valued Shepp-Logan phantom shown in \figref{cameraman}, and we constructed compressed noisy measurements $\vec{y}$ using spread-spectrum $\vec{\Phi}$ and AWGN $\vec{w}$ at a measurement SNR of $30$~dB in the Cameraman case and $40$~dB in the Shepp-Logan case.

\twoFrag{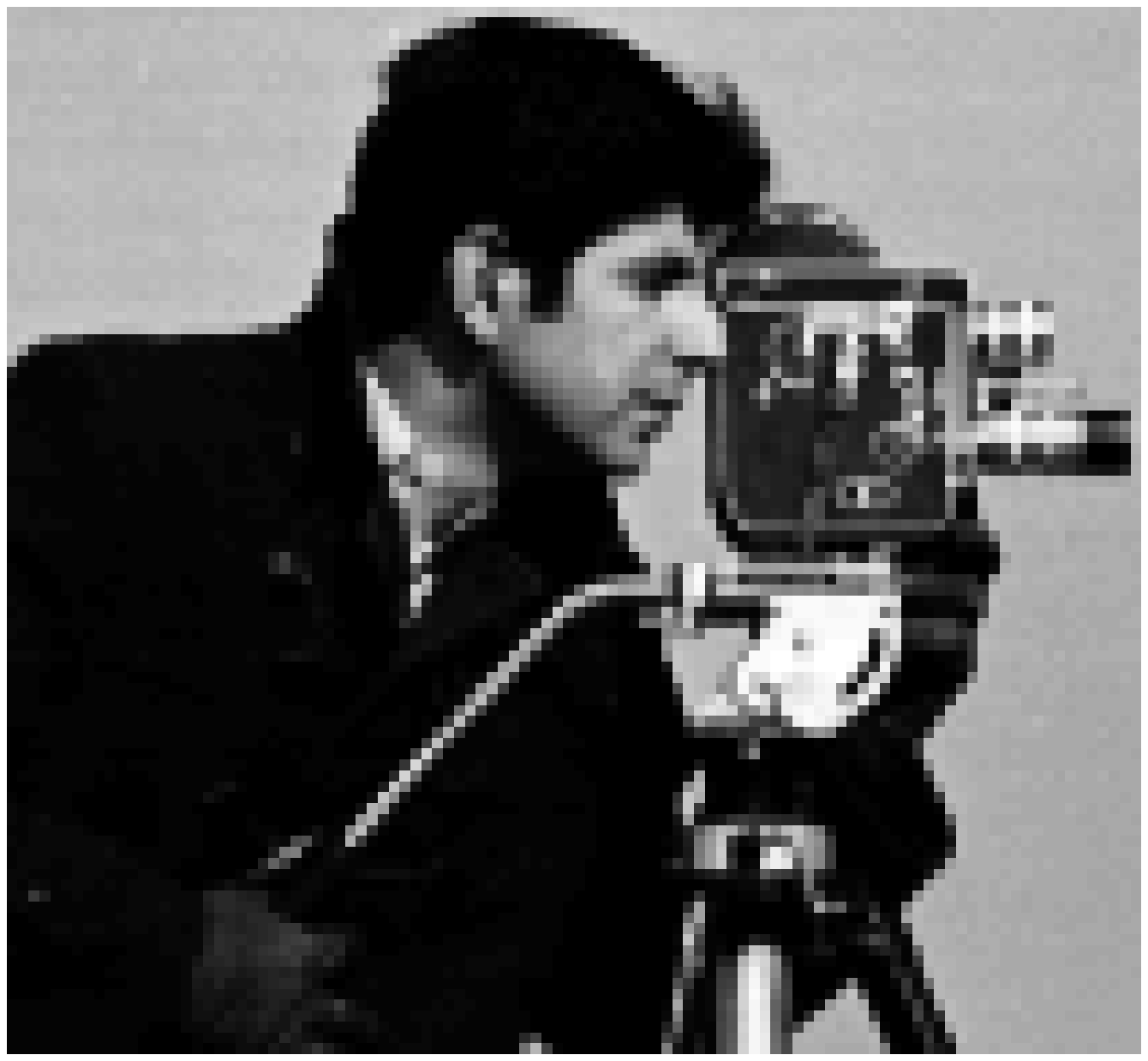}{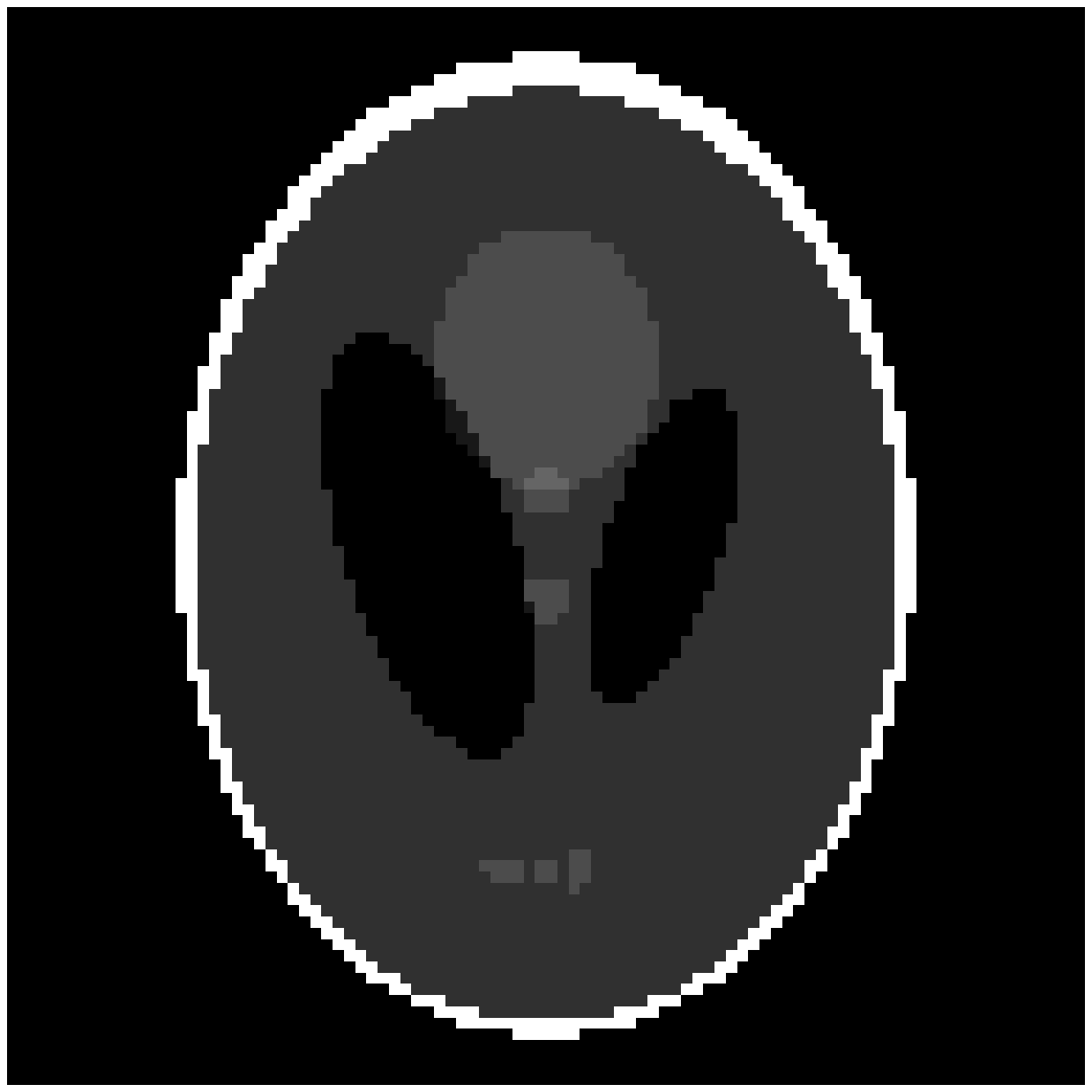} {0.422} {0.39}
        {Left: the real-valued cropped Cameraman image of size $N=96\times 104$. Right: the complex-valued Shepp-Logan phantom of size $N=96\times 96$. For the Shepp-Logan phantom, the real and imaginary parts of $\vec{x}$ were identical, and only the real part is shown here.} {}

For the Cameraman image, we constructed the analysis operator $\vec{\Psi}\in\Real^{8N\times N}$ by concatenating undecimated db1 and db2 2D wavelet transforms (UWT-db1-db2) with one level of decomposition. 
For the Shepp-Logan phantom image, we constructed the analysis operator $\vec{\Psi}\in\Real^{4N\times N}$ from the undecimated db1 2D wavelet transform (UWT-db1) with one level of decomposition.
The \score and \ascoreirw algorithms treated each of the sub-bands of the wavelet transform as a separate sub-dictionary $\vec{\Psi}_d$ in their composite regularizers.

\figref{cameraRmse} shows recovery SNR versus sampling ratio $M/N$ for the Cameraman image, while \figref{sheppRmse} shows the same for the Shepp-Logan phantom. 
Each recovery SNR represents the median value from $7$ independent realizations of $(\vec{\Phi},\vec{w})$. 
Both figures show that \score and \ascoreirw outperform their non-composite counterparts, especially at low sampling ratios; the gap between \ascoreirw and and \IRW closes at $M/N\geq 0.35$ for the Shepp-Logan phantom. 

\putFrag{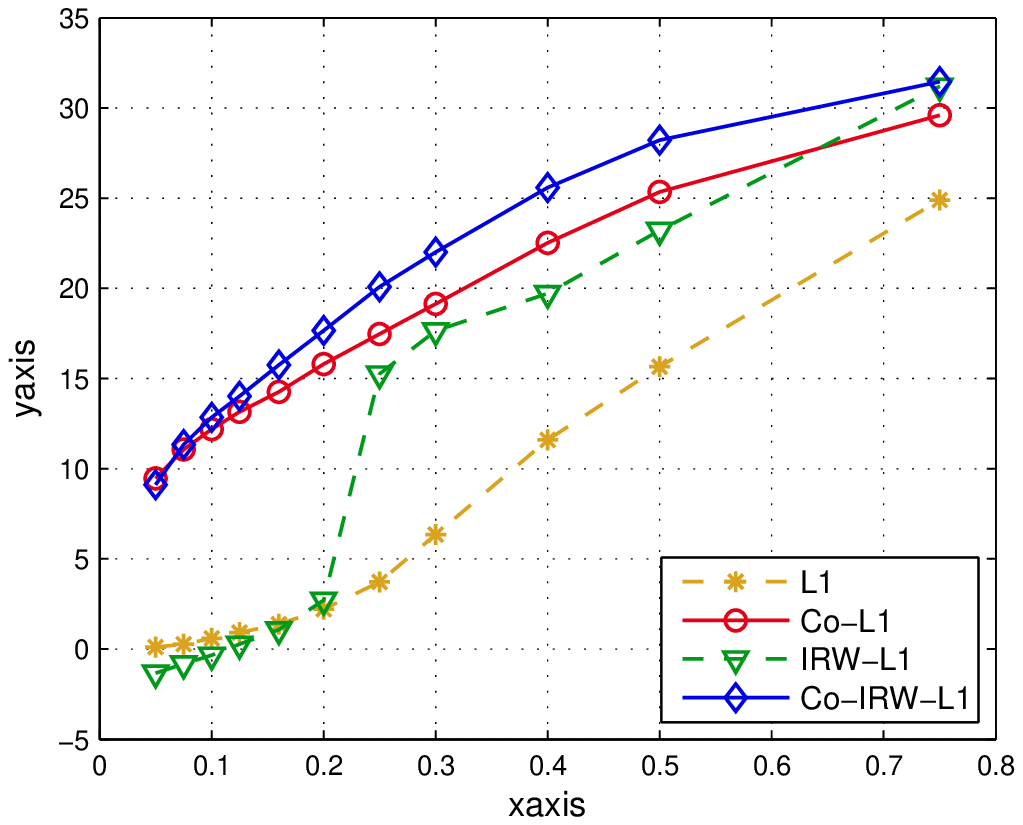}
        {Recovery SNR versus sampling ratio $M/N$ for the cropped Cameraman image.  Measurements were constructed using a spread-spectrum operator and AWGN at $30$~dB SNR, and recovery used UWT-db1-db2 at one level of decomposition. Each SNR value represents the median value from $7$ independent trials.}
{\figsize}
{\newcommand{\sz}{0.9}
 \newcommand{\szz}{0.75}
 \psfrag{yaxis}[b][b][\sz]{\sf recovery SNR [dB]}
 \psfrag{xaxis}[t][t][\sz]{\sf sampling ratio $M/N$}}

\putFrag{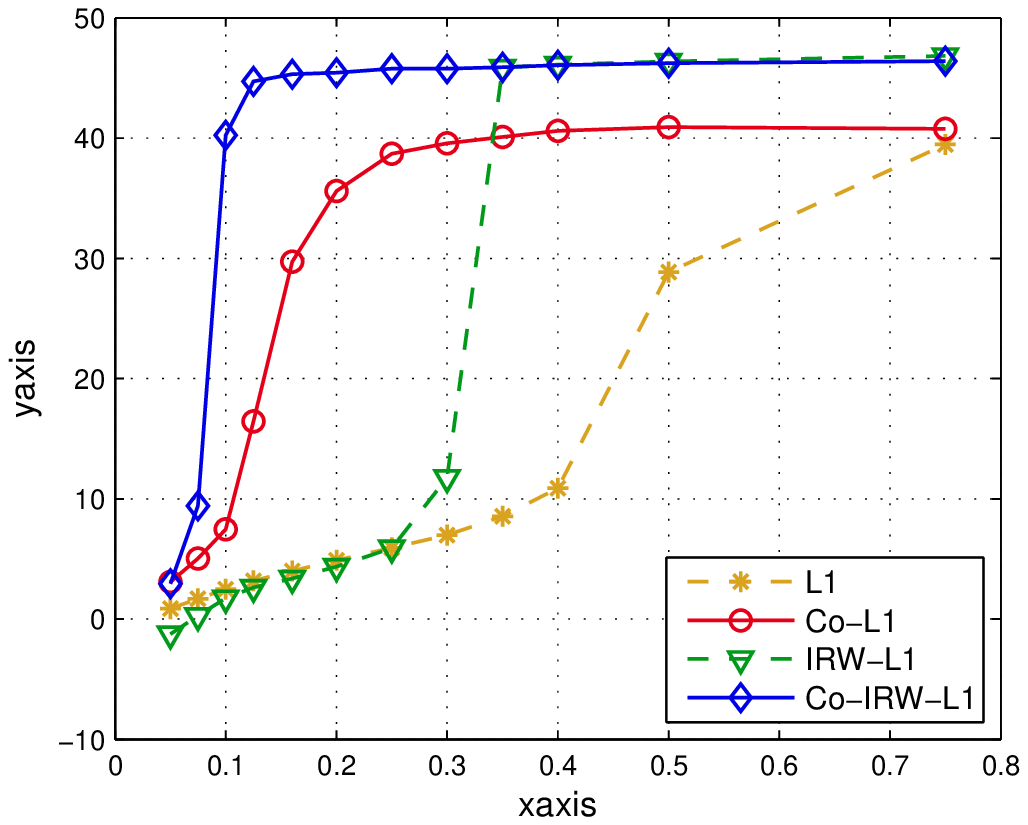}
        {Recovery SNR versus sampling ratio $M/N$ for the Shepp-Logan phantom. Measurements were constructed using a spread-spectrum operator and AWGN at $40$~dB SNR, and recovery used UWT-db1 at one level of decomposition.  Each recovery SNR represents the median value from $7$ independent trials.}
{\figsize}
{\newcommand{\sz}{0.9}
 \newcommand{\szz}{0.75}
 \psfrag{yaxis}[b][b][\sz]{\sf recovery SNR [dB]}
 \psfrag{xaxis}[t][t][\sz]{\sf sampling ratio $M/N$}}

\subsection{Dynamic MRI} \label{sec:mriRes}

For our third experiment, we investigate a simplified version of the ``dynamic MRI'' (dMRI) problem.
In dMRI, one attempts to recover a sequence of MRI images, known as an MRI cine, from highly under-sampled ``\textsf{k-t}-domain'' measurements $\{\vec{y}_t\}_{t=1}^T$ constructed as 
\begin{align}
\vec{y}_t
= \vec{\Phi}_t \vec{x}_t + \vec{w}_t ,
\label{eq:yt}
\end{align}
where $\vec{x}_t\in\Real^{N_1 N_2}$ is a vectorized ($N_1\times N_2$)-pixel image at time $t$, $\vec{\Phi}_t\in\Real^{M_1\times N_1 N_2}$ is a sub-sampled Fourier operator at time $t$, and $\vec{w}_t\in\Real^{M_1}$ is AWGN. 
This real-valued $\vec{\Phi}_t$ is constructed from the complex-valued $N_1 N_2 \times N_1 N_2$ 2D DFT matrix by randomly selecting $0.5 M_1$ rows and then splitting each of those rows into its real and imaginary components. 
Here, it is usually advantageous to vary the sampling pattern with time and to sample more densely at low frequencies, where most of the signal energy lies (e.g., \cite{Ahmad_2014_VISTA}).
Putting \eqref{yt} into the form of our measurement model \eqref{y}, we get 
\begin{align}
\underbrace{ 
\mat{\vec{y}_1\\[-1mm]\vdots\\\vec{y}_T} 
}_{\displaystyle \vec{y}}
= 
\underbrace{ 
\mat{\vec{\Phi}_1\\[-1mm]&\ddots\\&&\vec{\Phi}_T} 
}_{\displaystyle \vec{\Phi}}
\underbrace{ 
\mat{\vec{x}_1\\[-1mm]\vdots\\\vec{x}_T}
}_{\displaystyle \vec{x}}
+ 
\underbrace{ 
\mat{\vec{w}_1\\[-1mm]\vdots\\\vec{w}_T}
}_{\displaystyle \vec{w}} ,
\end{align}
with total measurement dimension $M=M_1 T$ and total signal dimension $N=N_1 N_2 T$.

As ground truth, we used a high-quality dMRI cardiac cine $\vec{x}$ of dimensions $N_1=144$, $N_2=85$, and $T=48$.
The left pane in \figref{mri2d} shows a $144\times 85$ image from this cine extracted at a single time $t$, while the middle pane shows a $144\times 48$ spatio-temporal profile from this cine extracted at a single horizontal location.
This middle pane shows that the temporal dimension is much more structured than the spatial dimension, suggesting that there may be an advantage to weighting the spatial and temporal dimensions differently in a composite regularizer.

\threeFrag{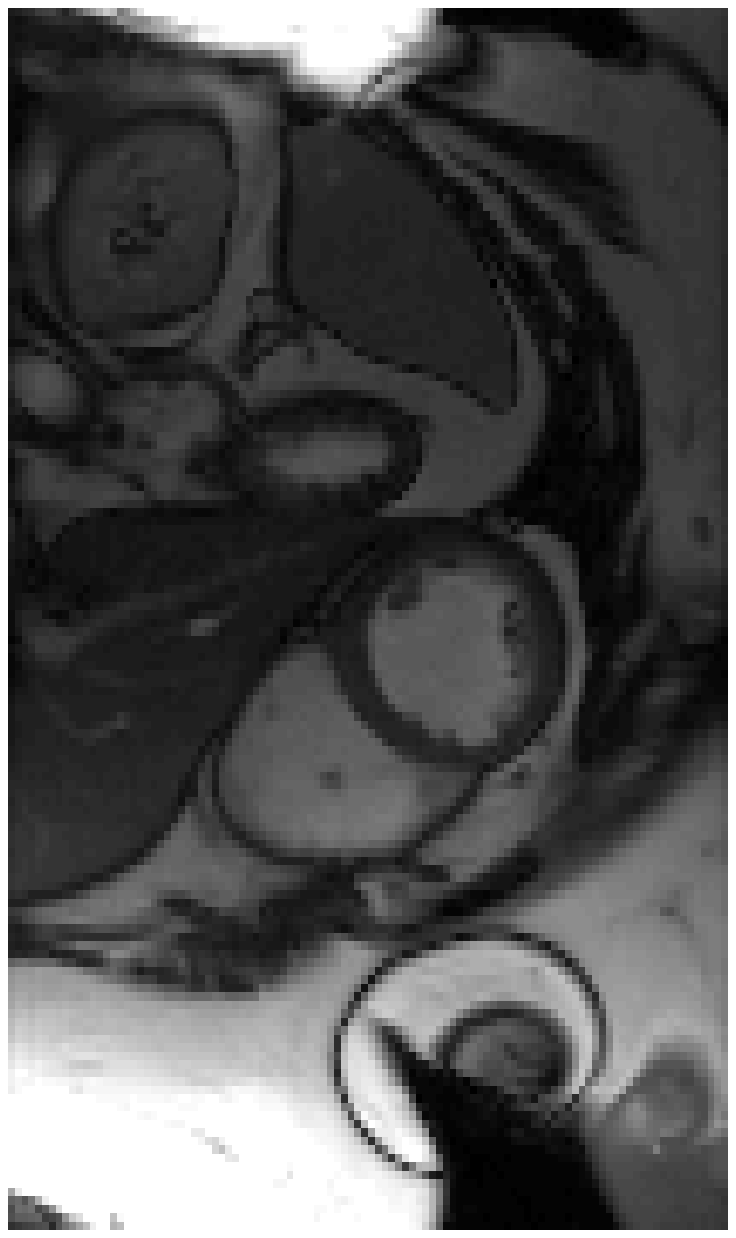}{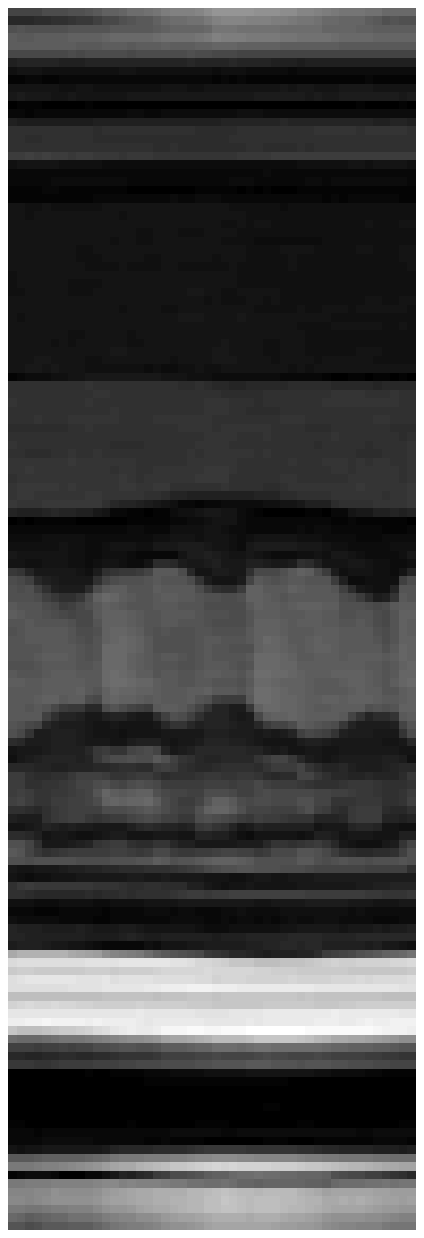}{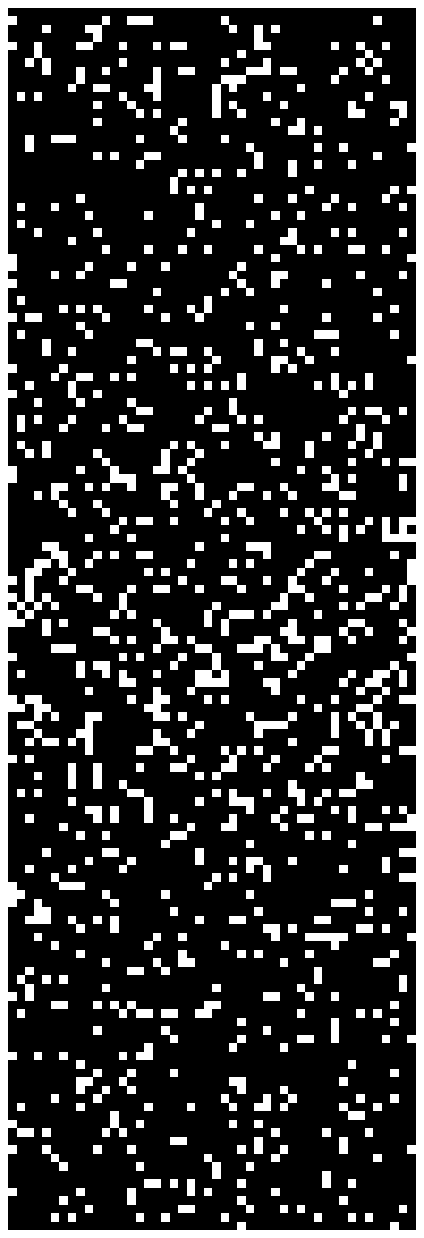} {0.366} {0.21} {0.21}
        {Left: A $144\times 85$ spatial slice from the $144\times 85\times 48$ dMRI dataset. Middle: The $144\times 48$ spatio-temporal slice used for the dMRI experiment.  Right: a realization of the variable-density \textsf{k}-space sampling pattern, versus time, at $M/N=0.30$.} {}

To test this hypothesis, we constructed an experiment where the goal was to recover the $144\times 48$ spatio-temporal profile shown in the middle pane of \figref{mri2d}, as opposed to the full 3D cine, from subsampled \textsf{k-t}-domain measurements.
For this purpose, we constructed measurements $\{\vec{y}\}_{t=1}^T$ as described above, but with $N_2=1$ (and thus a 1D DFT), and used a variable density random sampling method.
The right pane of \figref{mri2d} shows a typical realization of the sampling pattern versus time.
Finally, we selected the AWGN variance that yielded measurement SNR $=30$~dB. 

For the non-composite L1 and \IRW algorithms, we constructed the analysis operator $\vec{\Psi}\in\Real^{3N\times N}$ from a vertical concatenation of the db1-db3 orthogonal 2D discrete wavelet bases, each with two levels of decomposition.
For the \score and \ascoreirw algorithms, we assigned each of the 21 sub-bands in $\vec{\Psi}$ to a separate sub-dictionary $\vec{\Psi}_d\in\Real^{L_d\times N}$.
Note that the sub-dictionary size $L_d$ decreases with the level in the decomposition.
By weighting certain sub-dictionaries differently than others, the composite regularizers can exploit differences in spatial versus temporal structure.

\figref{mriRmse} shows recovery SNR versus sampling ratio $M/N$ for the four algorithms under test.
Each reported SNR represents the median SNR from $7$ independent realizations of $(\vec{\Phi},\vec{w})$.
The figure shows that \score outperforms its non-composite counterparts at all tested values of $M/N$, while \ascoreirw outperforms its noncomposite counterpart for $M/N\leq 0.4$. 
Although not shown here, we obtained similar results with other cine datasets and with an UWT-db1-based analysis operator.

\putFrag{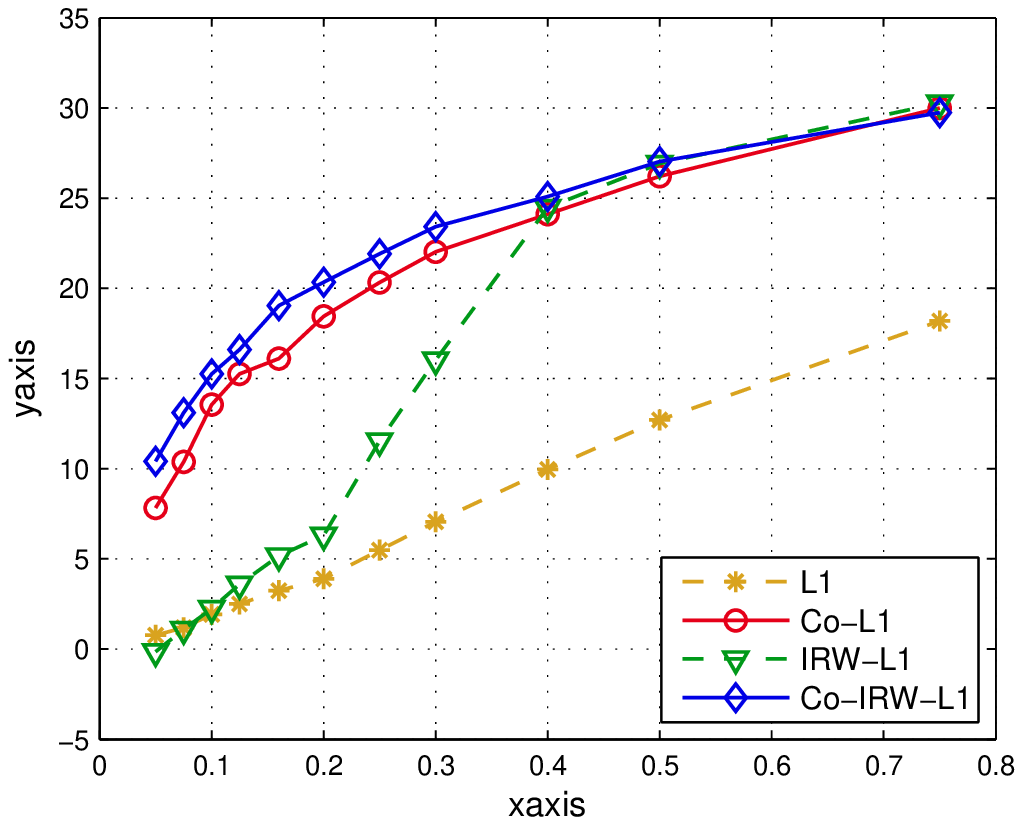}
{Recovery SNR versus sampling ratio $M/N$ for the dMRI experiment. Each SNR value represents the median value from $7$ independent trials. Measurements were constructed using variable-density sub-sampled Fourier operator and AWGN at $30$~dB measurement SNR, and recovery used a concatenation of db1-db3 orthogonal 2D wavelet bases at two levels of decomposition. 
}
{\figsize}
{\newcommand{\sz}{0.9}
 \newcommand{\szz}{0.75}
 \psfrag{yaxis}[b][b][\sz]{\sf recovery SNR [dB]}
 \psfrag{xaxis}[t][t][\sz]{\sf sampling ratio $M/N$}}

For qualitative comparison, \figref{mriRecon} shows the spatio-temporal profile recovered by each of the four algorithms under test at $M/N=0.3$ for a typical realization of $(\vec{\Phi},\vec{w})$.
Compared to the ground-truth profile shown in the middle pane of \figref{mri2d}, the profiles recovered by L1 and \IRW show visible artifacts that appear as vertical streaks.
In contrast, the profiles recovered by \score and \ascoreirw preserve most of the features present in the ground-truth profile. 

\begin{figure}
         \centering
         \captionsetup[subfigure]{labelformat=empty}
         \hfill
         \begin{subfigure}[t]{0.21\columnwidth}
                 \includegraphics[width=\columnwidth]{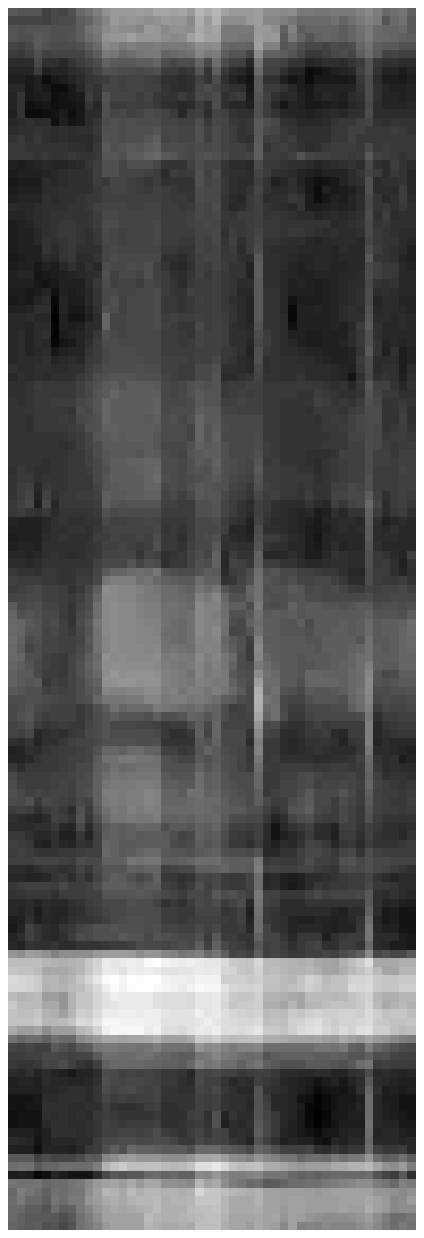}
                 \caption{\small \sf L1}
         \end{subfigure}%
         \hfill
         \begin{subfigure}[t]{0.21\columnwidth}
                 \includegraphics[width=\columnwidth]{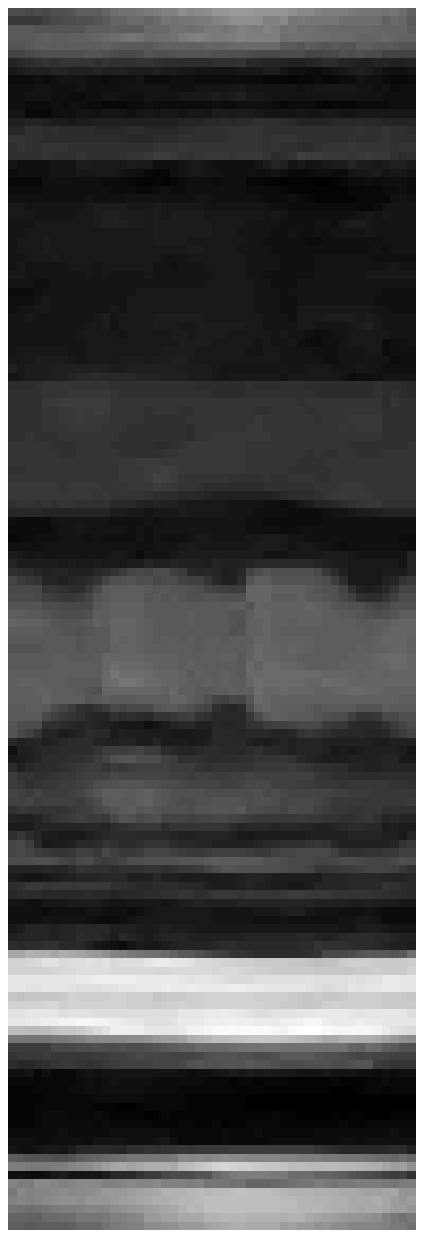}
                 \caption{\small \sf \score}
         \end{subfigure}
         \hfill
         \begin{subfigure}[t]{0.21\columnwidth}
                 \includegraphics[width=\columnwidth]{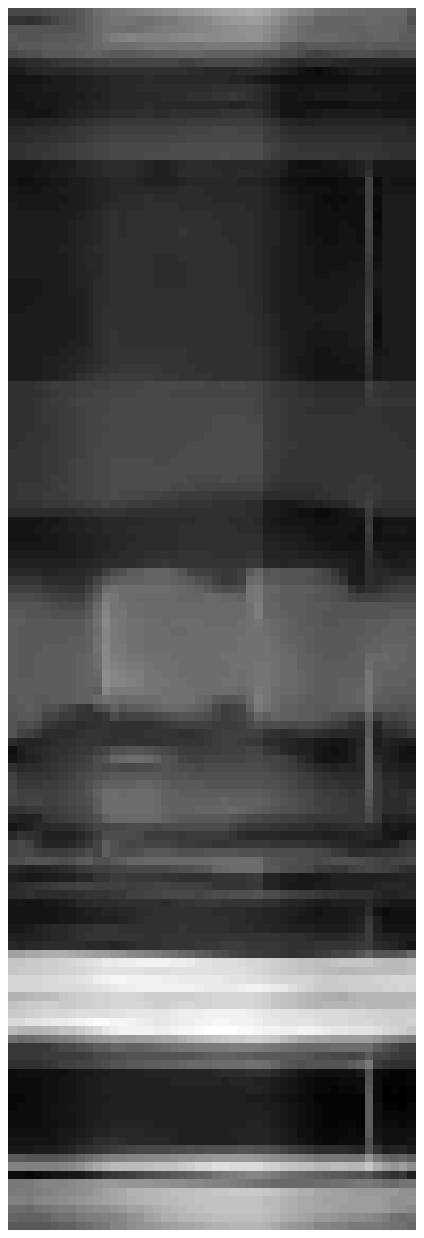}
                 \caption{\small \sf \IRW}
         \end{subfigure}
         \hfill
         \begin{subfigure}[t]{0.21\columnwidth}
                 \includegraphics[width=\columnwidth]{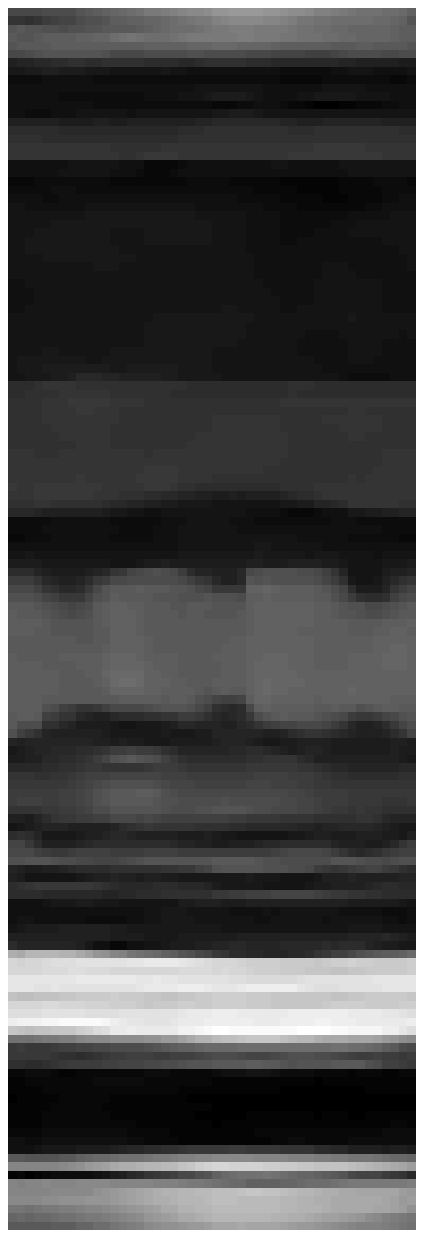}
                 \caption{\small \sf \ascoreirw}
         \end{subfigure}
         \hfill
         \caption{\small Recovered dMRI spatio-temporal profiles at $M/N=0.30$}
         \label{fig:mriRecon}
\end{figure}

\subsection{Algorithm Runtime} \label{sec:runtime}

\tabref{runtime} reports the average runtimes of the L1, \score, \IRW, and \ascoreirw algorithms for the experiments in Sections~\ref{sec:sl}~and~\ref{sec:mriRes}. 
There we see that the runtime of \score was $1.29\times$ that of L1 for the worst case, and the runtime of \ascoreirw was $1.33\times$ that of \IRW for the worst case.

\putTable{runtime}{Computation times (in seconds) for the presented experimental studies. The times are averaged over trial runs and different sampling ratios.}
{\begin{tabular}{|l|c|c|c|} \hline
  		    & Shepp-Logan 	& Cameraman 	& MRI 		\\ \hline
L1 		    & 8.12 		& 9.88		& 22.0   	\\ \hline
\score  	& 8.83 		& 12.8		& 21.7   	\\ \hline
\IRW 		& 7.95 		& 12.7		& 24.1  	\\ \hline
\ascoreirw 	& 9.29		& 16.9		& 29.6    	\\ \hline
\end{tabular}
}

\subsection{Choice of Dictionary} \label{sec:dictionary}
In our last experiment, we investigate the performance of \ascoreirw versus choice of $\{\vec{\Psi}_d\}$.
For this, we constructed $\{\vec{\Psi}_d\}$ using a concatenation of either undecimated or orthogonal 2D Daubechies wavelet transforms, and we varied both the number of transforms in the concatenation as well as the number of levels in the wavelet decomposition.
We then attempted to recover the Cameraman image from spread-spectrum measurements at $M/N=0.4$ in AWGN at $30$~dB SNR.
As usual, the \ascoreirw algorithm treated each wavelet sub-band as a separate sub-dictionary.

The recovery SNR for various choices of $\vec{\Psi}$ is shown in \figref{cameraDictRmse}. 
For the case of orthogonal wavelet transforms (OWT), a significant performance improvement was observed in going from one to two transforms, regardless of the wavelet decomposition level.
However, a slight performance degradation was observed when concatenating more than two OTWs. 
Moreover, the effect of varying the level of decomposition was mild unless no concatenation (i.e., db1) was used.
For the undecimated wavelet transform (UWT) case, the recovery SNR was essentially invariant to both the level of decomposition and the number of concatenated transforms, with only a slight degradation when five transforms were concatenated.
Overall, the UWT performed significantly better than the OWT. 
Similar trends were observed for the \score algorithm in experiments not shown here.  

\putFrag{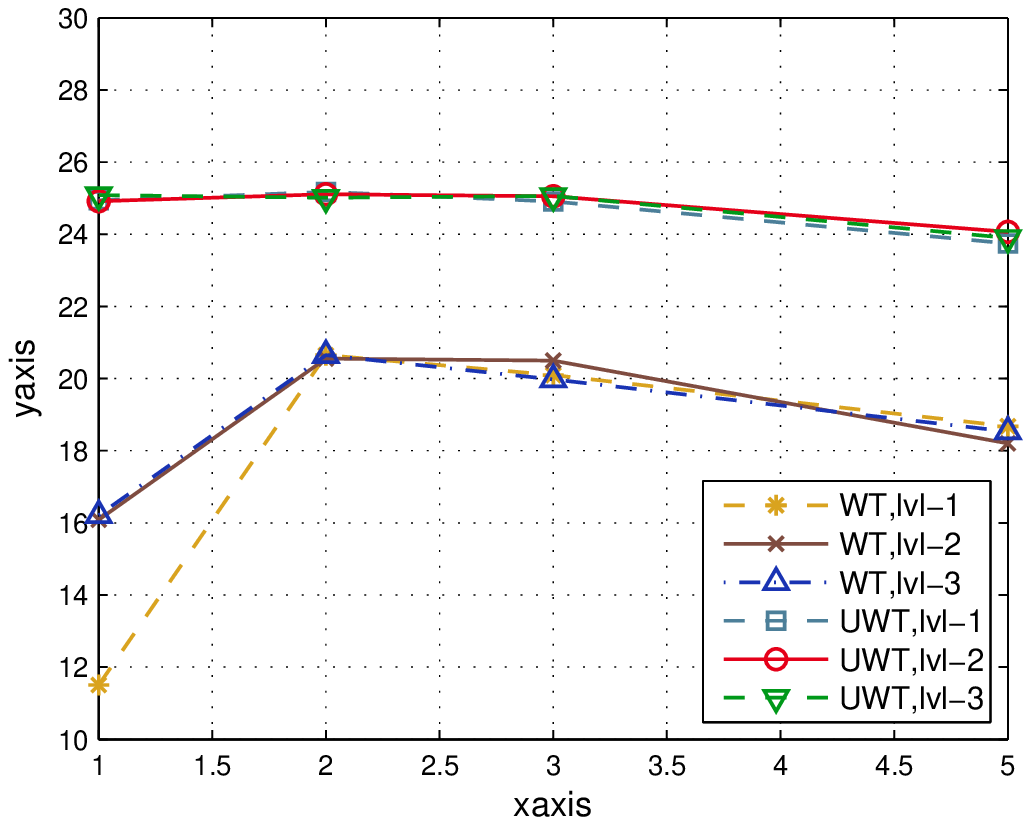}
        {\ascoreirw recovery SNR for different choices of $\vec{\Psi}_d$. Measurements were constructed from the cropped cameraman image using a spread-spectrum operator, AWGN at $30$~dB SNR, and sampling ratio $M/N=0.40$. Here, OWT represents a concatenation of 2D orthogonal Daubechies wavelet transforms, UWT represents a concatenation of 2D undecimated Daubechies wavelet transforms, and ``lvl'' denotes the level of decomposition.  Each SNR value represents the median value from $3$ independent trials.}
{\figsize}
{\newcommand{\sz}{0.9}
 \newcommand{\szz}{0.72}
 \psfrag{WT,lvl-1}[l][l][\szz]{\sf 1-lvl OWT}
 \psfrag{WT,lvl-2}[l][l][\szz]{\sf 2-lvl OWT}
 \psfrag{WT,lvl-3}[l][l][\szz]{\sf 3-lvl OWT}
 \psfrag{UWT,lvl-1}[l][l][\szz]{\sf 1-lvl UWT}
 \psfrag{UWT,lvl-2}[l][l][\szz]{\sf 2-lvl UWT}
 \psfrag{UWT,lvl-3}[l][l][\szz]{\sf 3-lvl UWT}
 \psfrag{yaxis}[b][b][\sz]{\sf recovery SNR [dB]}
 \psfrag{xaxis}[t][t][\sz]{\sf \# concatenated transforms}}

\section{Conclusions} \label{sec:con}

Motivated by the observation that a given signal $\vec{x}$ admits sparse representations in multiple dictionaries $\vec{\Psi}_d$ but with varying levels of sparsity across dictionaries, we proposed two new algorithms for the reconstruction of (approximately) sparse signals from noisy linear measurements.
Our first algorithm, \score, extends the well-known lasso algorithm \cite{Tibshirani:JRSSb:96,Chen:JSC:98,Tibshirani:AS:11} from the L1 penalty $\|\vec{\Psi x}\|_1$ to composite L1 penalties of the form \eqref{Rell} while self-adjusting the regularization weights $\lambda_d$.
Our second algorithm, \ascoreirw, extends the well-known IRW-L1 algorithm \cite{Figueiredo:TIP:07,Candes:JFA:08,Wipf:JSTSP:10} to the same family of composite penalties while self-adjusting the regularization weights $\lambda_d$ and the regularization parameters $\epsilon_d$.

We provided several interpretations of both algorithms: 
i) majorization-minimization (MM) applied to a non-convex log-sum-type penalty,
ii) MM applied to an approximate $\ell_0$-type penalty,
iii) MM applied to Bayesian MAP inference under a particular hierarchical prior, and 
iv) variational expectation-maximization (VEM) under a particular prior with deterministic unknown parameters.
Also, we leveraged the MM interpretation to establish convergence in the form of an asymptotic stationary point condition \cite{Mairal:ICML:13}.
Furthermore, we noted that the Bayesian MAP and VEM viewpoints yield novel interpretations of the original IRW-L1 algorithm.
Finally, we present a detailed numerical study that suggests that our proposed algorithms yield significantly improved recovery SNR when compared to their non-composite L1 and IRW-L1 counterparts with a modest (e.g., $1.3\times$) increase in runtime.

\section{Acknowledgment}

The authors thank the anonymous reviewers for their valuable feedback.

\appendices
\section{Lipschitz continuity of \score gradient} \label{app:score_lipschitz}

In this appendix, we establish the Lipschitz continuity of $\nabla g_2$ from \eqref{score_grad} in the case that $\epsilon>0$.
We first recall that, for $\nabla g_2$ to be Lipschitz continuous over the domain $\vec{v}\in\mc{C}$, there must exist some constant $\beta$ such that, for all $\vec{v},\vec{v}'\in\mc{C}$,
\begin{align}
\|\nabla g_2(\vec{v})-\nabla g_2(\vec{v}') \|_2^2 
&\leq \beta \|\vec{v}-\vec{v}'\|_2^2
\label{eq:lipschitz}
\end{align}
From \eqref{score_grad}, we have 
\begin{align}
\lefteqn{ 
\|\nabla g_2(\vec{v})-\nabla g_2(\vec{v}') \|_2^2 
}\nonumber\\
&= \sum_{k=1}^L \bigg(  
 \frac{L_{d(k)}}{\epsilon + \sum_{i\in\mc{K}_{d(k)}} v_{i}} 
-\frac{L_{d(k)}}{\epsilon + \sum_{i\in\mc{K}_{d(k)}} v_{i}'} 
\bigg)^2 \\
&= \sum_{k=1}^L 
 \frac{L_{d(k)}^2 \big[\sum_{i\in\mc{K}_{d(k)}} (v_{i}'-v_{i})\big]^2}
        {\big(\epsilon + \sum_{i\in\mc{K}_{d(k)}} v_{i}\big)^2
         \big(\epsilon + \sum_{i\in\mc{K}_{d(k)}} v_{i}'\big)^2} \\
&= \sum_{d=1}^{D}\sum_{l=1}^{L_d}
 \frac{L_{d}^2 \big[\sum_{i=1}^{L_d} (u_{d,i}'-u_{d,i})\big]^2}
        {\big(\epsilon + \sum_{i\in\mc{K}_{d(k)}} v_{i}\big)^2
         \big(\epsilon + \sum_{i\in\mc{K}_{d(k)}} v_{i}'\big)^2} .
\end{align}
We can then upper bound the latter as follows.
\begin{align}
\|\nabla g_2(\vec{v})-\nabla g_2(\vec{v}') \|_2^2 
&\leq \sum_{d=1}^{D}\sum_{l=1}^{L_d}
 \frac{L_{d}^2}{\epsilon^4} \bigg[\sum_{i=1}^{L_d} (u_{d,i}'-u_{d,i})\bigg]^2 
 \label{eq:score_lip1}\\
&\leq \sum_{d=1}^{D}
 \frac{L_{d}^3}{\epsilon^4} \bigg[\sum_{i=1}^{L_d} |u_{d,i}'-u_{d,i}|\bigg]^2\\
&\leq \sum_{d=1}^{D}
 \frac{L_{d}^4}{\epsilon^4} \sum_{i=1}^{L_d} (u_{d,i}'-u_{d,i})^2 
 \label{eq:score_lip2} \\
&\leq \frac{L_{\max}^4}{\epsilon^4} \sum_{k=1}^L (v_{k}'-v_{k})^2 
 \label{eq:score_lip3} \\
&\leq \frac{L_{\max}^4}{\epsilon^4} \sum_{k=1}^{L+N} (v_{k}'-v_{k})^2 \\
&= \frac{L_{\max}^4}{\epsilon^4} \|\vec{v}-\vec{v}'\|_2^2 ,
 \label{eq:score_lip} 
\end{align}
where 
\eqref{score_lip1} follows from the fact that $u_{d,l}\geq 0~\forall d,l$ (according to \eqref{score_logsum2}),
\eqref{score_lip2} follows from the fact that $\|\vec{x}\|_1\leq \sqrt{N} \|\vec{x}\|_2$ for $\vec{x}\in\Complex^N$,
and \eqref{score_lip3} uses $L_{\max}\defn \max_{d}L_d$.
Comparing \eqref{score_lip} to \eqref{lipschitz}, we see that $\nabla g_2$ from \eqref{score_grad} is Lipschitz continuous.

\section{Equivalence of Log-Sum and \texorpdfstring{$\ell_0$}{L0} Minimization} \label{app:l0}

In this appendix, we establish that the log-sum optimization \eqref{logsum} becomes equivalent to the $\ell_0$ optimization \eqref{l0} as $\epsilon\rightarrow 0$.
We first note that,
for any $\epsilon>0$,
\begin{align}
\lefteqn{ \frac{1}{\log(1/\epsilon)} \sum_{n=1}^N \log(\epsilon+|x_n|) }\\
&= \frac{1}{\log(1/\epsilon)} \left[ \sum_{n:\, x_n=0} \log(\epsilon)
        + \sum_{n:\,x_n\neq 0} \log(\epsilon+|x_n|) \right] \\
&= \|\vec{x}\|_0 - N
        + \frac{ \sum_{n:\,x_n\neq 0} \log(\epsilon+|x_n|)}{\log(1/\epsilon)} ,
\end{align}
where $\|\vec{x}\|_0$ is defined as the counting norm, i.e., $\|\vec{x}\|_0\defn|\{x_n: x_n\neq 0\}|$.
Applying this result to the objective function in \eqref{logsum}, we have
\begin{align}
\lefteqn{
\gamma \|\vec{y}-\vec{\Phi x}\|_2^2 + \sum_{n=1}^N \log(\epsilon+|x_n|) 
}\nonumber \\
&\propto 
        \underbrace{ \frac{\gamma}{\log(1/\epsilon)} }_{\displaystyle \defn \gamma'} \|\vec{y}-\vec{\Phi x}\|_2^2 + 
        \|\vec{x}\|_0 - N
        + \frac{ \displaystyle \sum_{n:\,x_n\neq 0} \log(\epsilon+|x_n|)}{\log(1/\epsilon)} 
\label{eq:scaled_logsum} .
\end{align}
Clearly the global scaling and offset by $N$ in \eqref{scaled_logsum} are inconsequential to the minimization in \eqref{logsum}.
Furthermore, by making $\epsilon>0$ arbitrarily small, we can make the last term in \eqref{scaled_logsum} arbitrarily small\footnote{Note that, as $\epsilon\rightarrow 0$, the numerator of the last term in \eqref{scaled_logsum} converges to the finite value $\sum_{n:\,x_n\neq 0} \log(|x_n|)$ while the denominator grows to $+\infty$.} and thus negligible compared to the other terms. 
It is in this sense that we say that \eqref{logsum} is equivalent to \eqref{l0} as $\epsilon\rightarrow 0$.

\section{Lipschitz continuity of \scoreirw gradient} \label{app:scoreirw_lipschitz}

In this appendix, we establish the Lipschitz continuity of $\nabla g_2$ from \eqref{scoreirw_grad} in the case that $\varepsilon>0$, recalling the Lipschitz definition \eqref{lipschitz}.
To ease the exposition, we focus on the $L=1$ case, noting that a similar (but more tedious) technique can be applied to the general case.

From the $L=1$ case of \eqref{scoreirw_grad}, we have
\begin{align}
\lefteqn{
|\nabla g_2(v)-\nabla g_2(v')|^2
}\nonumber\\
&= \bigg[ 
        \bigg(\frac{1}{\log(1+\varepsilon+\frac{v}{\epsilon_1})}+1\bigg)
        \frac{1}{\epsilon_1(1+\varepsilon)+v}
\nonumber\\&\quad
       -\bigg(\frac{1}{\log(1+\varepsilon+\frac{v'}{\epsilon_1})}+1\bigg)
        \frac{1}{\epsilon_1(1+\varepsilon)+v'}
\bigg]^2 \\
&= \big[A+B\big]^2 \\
&\leq \big[|A|+|B|\big]^2 \leq 2 \big[A^2+B^2\big] ,
\label{eq:grad_diff}
\end{align}
since $\|\vec{x}\|_1 \leq \sqrt{N}\|\vec{x}\|_2$ for $\vec{x}\in\Complex^N$, and where
\begin{align}
A
&\defn \frac{1}{\epsilon_1(1+\varepsilon)+v}
        -\frac{1}{\epsilon_1(1+\varepsilon)+v'} \\
B
&\defn \frac{1}{ (\epsilon_1(1+\varepsilon)+v)
        \log(1+\varepsilon+\frac{v}{\epsilon_1}) }
\nonumber\\&\quad
      -\frac{1}{ (\epsilon_1(1+\varepsilon)+v')
        \log(1+\varepsilon+\frac{v'}{\epsilon_1}) } .
\end{align}
Examining $A^2$, we find that
\begin{align}
A^2
&= \bigg( \frac{1}{\epsilon_1(1+\varepsilon)+v}
        -\frac{1}{\epsilon_1(1+\varepsilon)+v'} \bigg)^2 \\
&= \bigg( \frac{\epsilon_1(1+\varepsilon)+v' - [\epsilon_1(1+\varepsilon)+v]}
        {[\epsilon_1(1+\varepsilon)+v][\epsilon_1(1+\varepsilon)+v']}
        \bigg)^2 \\
&\leq (v'-v)^2/\epsilon_1^4
\label{eq:A}
\end{align}
since $\epsilon_1,\varepsilon>0$ and $v,v'\geq 0$.
Next, we write $B^2$ as 
\begin{align}
B^2
&= \frac{1}{\epsilon_1^2}\bigg( 
        \frac{1}{ \alpha \log(\alpha) }
       -\frac{1}{ \alpha' \log(\alpha') } 
   \bigg)^2 \\
&= \frac{1}{\epsilon_1^2} \bigg(
        \frac{\alpha' \log(\alpha')- \alpha\log(\alpha) }
        {\alpha \log(\alpha)\alpha' \log(\alpha')}
   \bigg)^2
\end{align}
with $\alpha\defn 1+\varepsilon+\frac{v}{\epsilon_1}$
and $\alpha' \defn 1+\varepsilon+\frac{v'}{\epsilon_1}$, and realize
\begin{align}
\lefteqn{ \alpha' \log(\alpha')- \alpha\log(\alpha) }\nonumber\\
&= (\alpha + \frac{v'-v}{\epsilon_1}) \log(\alpha')- \alpha\log(\alpha) \\
&= \alpha\log(\alpha') - \alpha\log(\alpha) + \frac{v'-v}{\epsilon_1}\log(\alpha')
\end{align}
which implies that
\begin{align}
B^2
&=\frac{1}{\epsilon_1^2}\bigg(
\underbrace{
\frac{1}{\alpha'\log(\alpha)}
-\frac{1}{\alpha'\log(\alpha')}
}_{\displaystyle \defn B_1}
+\underbrace{\frac{(v'-v)/\epsilon_1}{\alpha\alpha'\log(\alpha)}
}_{\displaystyle \defn B_2}
\bigg)^2 \\
&\leq \frac{\big[|B_1|+|B_2|\big]^2}{\epsilon_1^2} \leq \frac{2 \big[B_1^2+B_2^2\big]}{\epsilon_1^2} .
\label{eq:B}
\end{align}
Examining $B_1^2$ we find
\begin{align}
B_1^2
&= \frac{1}{\alpha'^2}\bigg( 
        \frac{1}{ \log(\alpha) }
       -\frac{1}{ \log(\alpha') } 
   \bigg)^2 \\
&= \frac{1}{\alpha'^2}\bigg( 
       \frac{\log(\alpha')-\log(\alpha)}{\log(\alpha)\log(\alpha')}
   \bigg)^2 \\
&= \frac{1}{\alpha'^2}
       \frac{\log(\alpha'/\alpha)^2}{\log(\alpha)^2\log(\alpha')^2}. 
\end{align}
Because $\epsilon_1,\varepsilon>0$ and $v,v'\geq 0$, we have that $\alpha,\alpha'>1$ and $\log(\alpha)^2\geq \log(1+\varepsilon)$ and $\log(\alpha')^2\geq \log(1+\varepsilon)$, so that
\begin{align}
B_1^2
\leq \frac{\log(\alpha'/\alpha)^2}{\log(1+\varepsilon)^4}. 
\label{eq:B1}
\end{align}
Moreover,
\begin{align}
\log(\alpha'/\alpha)^2
&= \log\bigg(\frac{\alpha+\frac{v'-v}{\epsilon_1}}{\alpha}\bigg)^2 \\
&= \log\bigg(1+\frac{v'-v}{\epsilon_1\alpha}\bigg)^2 \\
&\leq \max\Big\{\Big(\frac{v'-v}{\epsilon_1\alpha}\Big)^2 ,
                \Big(\frac{v'-v}{\epsilon_1\alpha + v'-v}\Big)^2 \Big\}
  \label{eq:log_alf1}\\
&= \frac{(v'-v)^2}{\epsilon_1^2}\max\Big\{\frac{1}{\alpha^2},
                                          \frac{1}{(\alpha')^2}\Big\}\\
&\leq \frac{(v'-v)^2}{\epsilon_1^2} 
  \label{eq:log_alf2},
\end{align}
where \eqref{log_alf1} used the property that $\frac{x}{1+x}\leq \log(1+x) \leq x$ for $x>-1$, and \eqref{log_alf2} used $\alpha,\alpha'>1$.
Finally, we have
\begin{align}
B_2^2 
&= \frac{(v'-v)^2}{\epsilon_1^2 (\alpha \alpha')^2 \log(1+\varepsilon+v/\epsilon_1)^2} \\
&\leq \frac{(v'-v)^2}{\epsilon_1^2 \log(1+\varepsilon)^2} 
\label{eq:B2}
\end{align}
where the latter step used $\alpha,\alpha'>1$ and $1+\varepsilon>0$ and $v/\epsilon_1\geq 0$.
Putting together \eqref{grad_diff}, \eqref{A}, \eqref{B}, \eqref{B1}, \eqref{log_alf2} and \eqref{B2}, we see that there exists $\beta>0$ such that
\begin{align}
|\nabla g_2(v)-\nabla g_2(v')|^2
\leq \beta (v'-v)^2 ~\forall (v',v)\in\mc{C},
\end{align}
implying that $\nabla g_2$ is Lipschitz continuous.

\bibliographystyle{IEEEtran}
\bibliography{col1,macros_abbrev,books,misc,sparse}
\end{document}